\documentclass{amsart}
\usepackage{amssymb}
\date{November 18, 2019}

\newtheorem{Thm}{Theorem}
\newtheorem{Cor}[Thm]{Corollary}
\newtheorem{Prop}[Thm]{Proposition}
\newtheorem{Lem}[Thm]{Lemma}

\theoremstyle{definition}

\theoremstyle{remark}

\newtheorem{Rem}[Thm]{Remark}

\def\td{\mathrm{d}}

\def\cz{\mathbb{C}}
\def\rz{\mathbb{R}}

\def\bs{\mathbb{S}}

\def\cd{\mathcal{D}}
\def\cE{\mathcal{E}}

\def\ci{\mathcal{I}}

\def\dk{\mathrm{d}k}

\def\dr{\mathrm{d}r}

\def\dt{\mathrm{d}t}

\def\dx{\mathrm{d}x}
\def\dy{\mathrm{d}y}

\def\domega{d\omega}
\def\gtf{\gamma_\mathrm{TF}}

\DeclareMathOperator{\ttr}{Tr}

\def\const{\mathrm{const}\,}

\newcommand{\ck}{\mathcal{K}}

\renewcommand{\epsilon}{\varepsilon}

\newcommand{\N}{\mathbb{N}}

\renewcommand{\phi}{\varphi}
\newcommand{\R}{\mathbb{R}}

\newcommand{\Sph}{\mathbb{S}}

\DeclareMathOperator{\ran}{ran}

\DeclareMathOperator{\spec}{spec}
\DeclareMathOperator{\Tr}{Tr}
\DeclareMathOperator{\tr}{Tr}

\begin{document}

\renewcommand{\dfrac}[2]{\frac{d#1}{d#2}}
\newcommand{\pfrac}[2]{\frac{\partial #1}{\partial #2}}
\newcommand{\me}[1]{\mathrm{e}^{#1}}
\newcommand{\one}{\mathbf{1}}

\title[Strong Scott Conjecture]{Proof of the Strong Scott Conjecture\\ for
  Chandrasekhar Atoms}

\dedicatory{Dedicated to Yakov Sinai on the occasion of his 85th birthday.}
 
\author{Rupert L. Frank}
\address[Rupert L. Frank]{Mathematisches Institut, Ludwig-Maximilans Universit\"at M\"unchen, Theresienstr. 39, 80333 M\"unchen, Germany, and Munich Center for Quantum Science and Technology (MCQST), Schellingstr. 4, 80799 M\"unchen, Germany, and Mathematics 253-37, Caltech, Pasa\-de\-na, CA 91125, USA}
\email{rlfrank@caltech.edu}

\author{Konstantin Merz}
\address[Konstantin Merz]{Mathematisches Institut, Ludwig-Maximilans Universit\"at M\"unchen, Theresienstr. 39, 80333 M\"unchen, Germany, and Munich Center for Quantum Science and Technology (MCQST), Schellingstr. 4, 80799 M\"unchen, Germany}
\email{merz@math.lmu.de}

\author{Heinz Siedentop}
\address[Heinz Siedentop]{Mathematisches Institut, Ludwig-Maximilans Universit\"at M\"unchen, Theresienstr. 39, 80333 M\"unchen, Germany, and Munich Center for Quantum Science and Technology (MCQST), Schellingstr. 4, 80799 M\"unchen, Germany}
\email{h.s@lmu.de}

\author[Barry Simon]{Barry Simon}
\address[Barry Simon]{Mathematics 253-37, Caltech, Pasa\-de\-na, CA 91125, USA}
\email{bsimon@caltech.edu}

\begin{abstract}
  We consider a large neutral atom of atomic number $Z$, taking
  relativistic effects into account by assuming the dispersion
  relation $\sqrt{c^2p^2+c^4}$.
  We study the behavior of the one-particle ground state density on
  the length scale $Z^{-1}$ in the limit $Z,c\to\infty$ keeping
  $Z/c$ fixed and find that the spherically averaged density as well as all individual
  angular momentum densities separately converge to the relativistic
  hydrogenic ones.
  This proves the generalization of the strong Scott conjecture for
  relativistic atoms and shows, in particular, that relativistic
  effects occur close to the nucleus. Along the way we prove upper bounds on the relativistic hydrogenic density.
\end{abstract}

\maketitle
\section{Introduction\label{s1}}
\subsection{Some results on large $Z$-atoms\label{ss1.1}}

The asymptotic behavior of the ground state energy and the ground state density of atoms with large atomic number $Z$ have been studied in detail in
non-relativistic quantum mechanics.

Soon after the advent of quantum mechanics it became clear that the
non-relativistic quantum multi-particle problem is not analytically
solvable and of increasing challenge with large particle number. This
problem was addressed by Thomas \cite{Thomas1927} and Fermi
\cite{Fermi1927,Fermi1928} by developing what was called the
statistical model of the atom. The model is described by the so-called
Thomas--Fermi functional (Lenz \cite{Lenz1932})
\begin{equation}
  \cE_Z^{\mathrm{TF}}(\rho) := \int_{\rz^3}\left(\tfrac{3}{10}\gtf\rho(x)^{5/3}-\frac{Z}{|x|}\rho(x)\right)\dx+\underbrace{\frac12\iint_{\rz^3\times\rz^3}\frac{\rho(x)\rho(y)}{|x-y|}\,\dx\,\dy}_{=:D[\rho]} \,,
\end{equation}
where $\gtf=(6\pi^2/q)^{2/3}$ is a positive constant depending on the
number $q$ of spin states per electron, i.e., physically $2$.
This functional is naturally defined on all densities with finite
kinetic energy and finite self-interaction, i.e., on
$$\ci:=\{\rho\geq0| \rho\in L^{5/3}(\rz^3)\cap D[\rho]<\infty\}.$$
The ground state energy in Thomas--Fermi theory is given by
$$
E^\mathrm{TF}(Z) := \inf_{\rho\in\ci} \cE_Z^{\mathrm{TF}}(\rho) \,.
$$
The Thomas--Fermi functional has a well known scaling behavior. Its
minimum scales as
$$ E^\mathrm{TF}(Z) = E^\mathrm{TF}(1)\ Z^{7/3}$$
and its minimizer as
$$\rho_Z^\mathrm{TF}(x) = Z^2\rho_1^\mathrm{TF}(Z^{1/3}x).$$
It is therefore natural to conjecture that the energy $E^{\rm S}(Z)$ of the
non-relativistic atomic Schr\"odinger operator and corresponding ground
state densities $\rho_Z^\mathrm{S}$ would -- suitably rescaled --
converge to the corresponding Thomas--Fermi quantity. In fact, fifty
years after the work of Thomas and Fermi, Lieb and Simon
\cite{LiebSimon1973,LiebSimon1977} showed
$$ E^{\rm S}(Z)= E^\mathrm{TF}(Z) + o(Z^{7/3})$$
and
$$ \int_M\dx\ \rho_Z^\mathrm{S}(x/Z^{1/3})/Z^2 \to\int_M\dx\ \rho_1^\mathrm{TF}(x)$$
for every bounded measurable set $M$ (see also Baumgartner
\cite{Baumgartner1976} for the convergence of the density).

Whereas the energy is asymptotically given by Thomas--Fermi theory, it
turns out that the relative accuracy for medium range atoms is only
about 10 \%; in fact, the energy given by the Thomas--Fermi approximation is too low and this triggered discussions
for corrections. Initially -- the result of Lieb and Simon was not yet
available -- it was proposed, e.g., to change the power, namely instead
of $E^\mathrm{TF}(1)Z^{7/3}$ to the dependence $\const Z^{12/5}$ (Foldy
\cite{Foldy1951}) motivated by numerical results.  An alternative
correction, namely to add $Z^2/2$, was put forward by Scott
\cite{Scott1952} based on a theoretical argument that a correction on
distances $1/Z$ from the nucleus is necessary. This prediction --
later named Scott conjecture (Lieb \cite{Lieb1981}, Simon
\cite[Problem 10b]{Simon1984}) -- was shown to be correct (Siedentop
and Weikard
\cite{SiedentopWeikard1987U,SiedentopWeikard1987U,SiedentopWeikard1988,SiedentopWeikard1989}
where the lower bound is based on initial work of Hughes
\cite{Hughes1986}, see also Hughes \cite{Hughes1990}). The related
statement that the density on the scale $1/Z$ converges to the
sum of the square of the hydrogen orbitals, also known as strong Scott
conjecture (Lieb \cite{Lieb1981}), was shown by Iantchenko et al
\cite{Iantchenkoetal1996}. We refer to Iantchenko \cite{Iantchenko1995}
for the density on other scales, and Iantchenko and Siedentop
\cite{IantchenkoSiedentop2001} for the one-particle density matrix.

In fact, Schwinger \cite{Schwinger1981} proposed a three term
expansion
\begin{equation}
  \label{eq:gse}
  E^{\rm S}(Z)=E^\mathrm{TF}(1) Z^{7/3}+ \tfrac q4Z^{2} - e_\mathrm{DS} Z^{5/3} + o(Z^{5/3})
\end{equation}
for the ground state energy, which was shown to be correct by Fefferman
and Seco
\cite{FeffermanSeco1990O,FeffermanSeco1992,FeffermanSeco1993,FeffermanSeco1994,FeffermanSeco1994T,FeffermanSeco1994Th,FeffermanSeco1995}.
A corresponding result for the density is still unknown.

Although these asymptotic expansions for large $Z$-atoms are extremely
challenging, they are of limited physical interest, since the
innermost electrons move with a velocity which is a substantial
fraction of the velocity of light $c$. Thus, a relativistic description is
mandatory.
% and should be visible in a lowering of Scott's correction
% which results just from the innermost electrons.
Here, we study one of the simplest relativistic models in quantum mechanics,
the so-called Chandrasekhar operator.
S\o rensen \cite{Sorensen2005} proved that for $Z\to\infty$ and $c\to\infty$, keeping the ratio $Z/c\leq2/\pi$ fixed, the ground state energy is again given to
leading order by $E^{\rm TF}(Z)$.
Moreover, it was shown \cite{MerzSiedentop2019} that the ground state
density on the Thomas--Fermi length scale converges weakly and in the
Coulomb norm to the minimizer of the hydrogenic Thomas--Fermi functional.
This indicates that the bulk of the electrons on the length scale
$Z^{-1/3}$ does not behave relativistically.
However, relativistic effects should be visible in the Scott correction
which results from the innermost and fast moving electrons.
In fact, Schwinger \cite{Schwinger1980} made this observation quantitative
and derived a $Z^2$ correction which is lower than Scott's. Such a
correction was indeed proven for the Chandrasekhar operator by Solovej et al
\cite{Solovejetal2008} and Frank et al \cite{Franketal2008} by two different
methods. Later, it was shown for other relativistic Hamiltonians, namely,
the Brown-Ravenhall operator (Frank et al \cite{Franketal2009}), and the
no-pair operator in the Furry picture (Handrek and Siedentop
\cite{HandrekSiedentop2015}), which describe increasingly more realistic
models. In fact, the no-pair operator in the Furry picture gives
numerically energies that are correct within chemical accuracy (Reiher and
Wolf \cite{ReiherWolf2009}).

Our main result here is the convergence of the suitably rescaled
one-particle ground state density of a Chandrasekhar atom: it converges
on distances $1/Z$ from the nucleus to the corresponding density of
the one-particle hydrogenic Chandrasekhar operator. This proves a
generalization of the strong Scott conjecture for relativistic atoms.
This result underscores that relativistic effects occur close to
the nucleus and that self-interactions of the innermost electrons are
negligible.

%%%%%%%%%%%

\subsection{Definitions and main result\label{Bezeichnung und Resultate}}

The Chandrasekhar operator of $N$ electrons of unit mass and with
$q$ spin states, each in the field of a nucleus of charge $Z$ and
velocity of light $c>0$, is given, in atomic units, by
\begin{equation}
  \sum_{\nu=1}^N\left(\sqrt{-c^2\Delta_\nu+c^4}-c^2-\frac{Z}{|x_\nu|}\right)
  +\sum_{1\leq \nu<\mu\leq N}\frac{1}{|x_\nu-x_\mu|}
  \quad\text{in}\ \bigwedge_{\nu=1}^N L^2(\rz^3:\cz^q) \,.
\end{equation}
Technically, this operator is defined as the Friedrichs extension of
the corresponding quadratic form with form domain
$\bigwedge_{\nu=1}^NC_0^\infty(\rz^3:\cz^q)$. It is bounded from below
if and only if $Z/c\leq2/\pi$ (Kato
\cite[Chapter 5, Equation (5.33)]{Kato1966}, Herbst
\cite[Theorem 2.5]{Herbst1977}, Weder \cite{Weder1974}) and, if
$Z/c<2/\pi$, then its form domain is
$H^{1/2}(\rz^{3N}:\cz^{q^N})\cap \bigwedge_{\nu=1}^N L^2(\rz^3:\cz^{q})$.
In the following, we restrict ourselves to the case where
$$
N=Z \,,\qquad q=1 \qquad\text{and}\qquad \gamma=Z/c
$$
with a fixed constant $\gamma\in (0,2/\pi)$. We denote the
resulting Hamiltonian by $C_Z$.

It is known that the ground state energy $\inf\spec C_Z$ is an
eigenvalue of $C_Z$ (Lewis et al \cite{Lewisetal1997}). This eigenvalue
may be degenerate and we denote by $\psi_1,\ldots,\psi_M$ a basis of the
corresponding eigenspace. In the following we consider (not necessarily
pure) ground states $d$ of $C_Z$, which can be written as
$$
d=\sum_{\mu=1}^M w_\mu\lvert\psi_\mu\rangle\langle\psi_\mu\rvert
$$
with constants $w_\mu\geq 0$ such that $\sum_{\mu=1}^M w_\mu=1$. We will
denote the corresponding one-particle density by $\rho_d$,
$$
\rho_d(x):=N\sum_{\mu=1}^Mw_\mu\int_{\rz^{3(N-1)}} |\psi_\mu(x,x_2,\ldots,x_N)|^2\,\dx_2 \cdots \dx_N
\qquad\text{for}\ x\in\R^3 \,.
$$
For $\ell\in\N_0$ we denote by $Y_{\ell m}$, $m=-\ell,\ldots,\ell$, a basis
of spherical harmonics of degree $\ell$, normalized in $L^2(\Sph^2)$ \cite[(B.93)]{Messiah1969}.
The radial electron density $\rho_{\ell,d}$ in the $\ell$-th angular
momentum channel will be denoted by $\rho_{\ell,d}$,
\begin{align*}
  &    \rho_{\ell,d}(r):=\frac{Nr^2}{2\ell+1} \!\sum_{m=-\ell}^\ell\sum_{\mu=1}^M w_\mu \!\! \int_{\rz^{3(N-1)}} \!\!\left|\int_{\bs^2}\overline{Y_{\ell m}(\omega)}\psi_\mu(r\omega,x_2,\ldots,x_N)\domega\right|^2 \dx_2\cdots\dx_N \\
  & \qquad    \qquad\text{for}\ r\in\R_+ \,.
\end{align*}
Note that
\begin{equation}
\label{eq:totalaverage}
\int_{\mathbb S^2} \rho_d(r\omega)\,{\rm d}\omega = r^{-2} \sum_{\ell=0}^\infty (2\ell+1) \rho_{\ell,d}(r)
\qquad\text{for}\ r\in\R_+ \,.
\end{equation}

Our main result concerns these densities on distances of order $Z^{-1}$ from the nucleus. It is known that electrons on these distances lead to the Scott correction to the Thomas--Fermi approximation to the ground state energy of $C_Z$; see \cite{Solovejetal2008,Franketal2008}. As in these works, a key role in our paper is played by the relativistic hydrogen
Hamiltonian
\begin{align*}
  C^H := \sqrt{-\Delta+1}-1 - \frac{\gamma}{|x|} \quad\text{in}\ L^2(\R^3) \,.
\end{align*}
Decomposing this operator into angular momentum channels we are led to the
radial operators
\begin{equation}
\label{eq:defclh}
C_\ell^H := \sqrt{-\frac{\mathrm d^2}{\mathrm d r^2}+\frac{\ell(\ell+1)}{r^2}+1}-1-\frac{\gamma}{r}
\quad\text{in}\ L^2(\R_+) \,.
\end{equation}
We emphasize that the space $L^2(\R_+)$ is defined with measure $\dr$, not
with $r^2\,\dr$. If $\psi_{n,\ell}^H$, $n\in\N_0$, denote the normalized
eigenfunctions of this operator, we denote the corresponding density in
channel $\ell$ by
\begin{equation}
  \label{eq:2.12}
  \rho_{\ell}^H(r):=\sum_{n=0}^\infty |\psi_{n,\ell}^H(r)|^2 \,.
\end{equation}
The total density is given by
\begin{equation}
\label{eq:totaldenshydro}
\rho^H(r):= (4\pi)^{-1} \, r^{-2} \sum_{\ell=0}^\infty (2\ell+1) \rho_{\ell}^H(r) \,.
\end{equation}
We discuss properties of these densities later in Theorem \ref{existencerhoh}, where we show, in particular, that the above series converge for $r>0$ and where we prove bounds on their small $r$ and large $r$ behavior.

The strong Scott conjecture asserts convergence of the rescaled ground state densities $\rho_d$ and $\rho_{\ell,d}$ to the corresponding relativistic hydrogen densities $\rho^H$ and $\rho_\ell^H$. This convergence holds in the weak sense when integrated against test functions. Our test functions are allowed to be rather singular at the origin and do not need to decay rapidly. Since the definition of the corresponding function spaces $\cd_{\gamma}^{(0)}$ and $\cd$ is somewhat involved, we do not state it here but refer to \eqref{eq:deftestnod} and \eqref{eq:deftest} in Sections~\ref{s:singlel} and \ref{s:alll}. We denote by $L^p_{\rm c}([0,\infty))$ the space of all functions in $L^p$ whose support is a compact subset of $[0,\infty)$. As example of allowed test function we mention that if
$$
|U(r)| \leq C \left( r^{-1} \one_{\{r\leq 1\}} + r^{-\alpha} \one_{\{r>1\}}\right)
$$
for some $\alpha>1$, then $U=U_1+U_2$ with $U_1\in r^{-1}L^\infty_{\mathrm c}([0,\infty))$ and $U_2\in\cd_\gamma^{(0)}$. Moreover, if $\alpha>3/2$, then even $U_2\in\cd\cap\cd_\gamma^{(0)}$. Note that if indeed, as we believe, $\rho^H(r)\gtrsim r^{-3/2} \one_{\{r>1\}}$, then the assumption $\alpha>3/2$ is optimal in order to have the integral $\int_{\R^3} \rho^H(|x|)U(|x|)\,\dx$, which appears in the strong Scott conjecture, finite.

The following two theorems are our main results.

\begin{Thm}[Convergence in a fixed angular momentum channel]
  \label{convfixedl}
  Let $\gamma\in(0,2/\pi)$, $\ell_0\in\N_0$, and $U=U_1+U_2$ with
  $U_1\in r^{-1}L^\infty_{\mathrm c}([0,\infty))$ and $U_2\in\cd_\gamma^{(0)}$.
  Then, for $Z,c\to\infty$ with $Z/c=\gamma$ fixed,
  \begin{align*}
    \lim_{Z\to\infty}\int_0^\infty c^{-3}\rho_{\ell_0,d}(c^{-1}r)U(r)\,\dr
    = \int_0^\infty \rho_{\ell_0}^H(r)U(r)\,\dr\,.
  \end{align*}
\end{Thm}

\begin{Thm}[Convergence in all angular momentum channels]
  \label{convalll}
  Let $\gamma\in(0,2/\pi)$ and $U=U_1+U_2$ be a function on $(0,\infty)$
  with $U_1\in r^{-1}L^\infty_{\mathrm c}([0,\infty))$ and $U_2\in\cd\cap\cd_{\gamma}^{(0)}$.
  Then, for $Z,c\to\infty$ with $Z/c=\gamma$ fixed,
  \begin{align*}
    \lim_{Z\to\infty}\int_{\R^3} c^{-3}\rho_{d}(c^{-1}|x|)U(|x|)\,\dx = \int_{\R^3} \rho^H(|x|)U(|x|)\,\dx\,.
  \end{align*}
\end{Thm}

As explained in the previous subsection, this is the relativistic analogue of
the strong Scott conjecture proved by Iantchenko et al
\cite{Iantchenkoetal1996}. Let us compare
Theorems~\ref{convfixedl} and~\ref{convalll} with their results. Both
works give convergence of the density on scales of order $Z^{-1}$ in a
certain weak sense. On the one hand, our class of test functions includes
functions with Coulomb type singularities (and even a behavior
like $\sin(1/r)r^{-1}$ is allowed) and with slow decay like $r^{-3/2-\epsilon}$ for $\epsilon>0$, which are not covered in
\cite{Iantchenkoetal1996}. On the other hand, the class of test functions
in \cite{Iantchenkoetal1996} includes radial delta functions and therefore
Iantchenko et al can prove pointwise convergence. The reason we cannot handle radial delta functions is that these are not form bounded with respect to $\sqrt{-\Delta}$.
% This is not expected in the relativistic
% case since the relativistic eigenfunctions can be singular.

\begin{Rem}\label{approxgs}
The fact that $d$ is an exact ground state of $C_Z$ is not essential for the proof of Theorems \ref{convfixedl} and \ref{convalll}. The assertions continue to hold if $d$ is an approximate ground state in the sense that
\begin{equation}
\label{eq:approxgs}
Z^{-2} \left( \tr C_Z d - \inf\spec C_Z \right) \to 0
\qquad\text{as}\ Z\to\infty \,;
\end{equation}
see Remark \ref{approxgs1} for further details about this generalization.
\end{Rem}

We end this presentation of our main results by discussing the hydrogenic densities $\rho^H$ and $\rho^H_\ell$ in more detail. They are much less understood than their non-relativistic counterparts. This originates from the fact that the eigenfunctions in the Chandrasekhar case are not explicitly known as opposed to the Schr\"odinger case.

The following result gives pointwise bounds on the densities $\rho_{\ell}^H$ and $\rho^H$ and shows, in particular, that the series defining them actually converge. To formulate the result, we introduce
\begin{align}
  \label{eq:defsigma}
  \sigma\mapsto\Phi(\sigma):=(1-\sigma)\tan\frac{\pi\sigma}{2}\,.
\end{align}
It is easy to see that $\Phi$ is strictly monotone on $[0,1]$ with
$\Phi(0)=0$ and $\lim_{\sigma\to1}\Phi(\sigma)=2/\pi$. Thus, there is a
unique $\sigma_\gamma\in (0,1)$ such that $\Phi(\sigma_\gamma)=\gamma$.

\begin{Thm}[Pointwise bounds on $\rho_\ell^H$ and $\rho^H$]
  \label{existencerhoh}
  Let $1/2<s\leq 3/4$ if $0<\gamma<(1+ \sqrt 2)/4$ and $1/2<s<3/2-\sigma_\gamma$ if $(1+\sqrt 2)/4\leq\gamma<2/\pi$. Then for all $\ell\in\N_0$ and $r\in\R_+$
  \begin{align*}
    \rho_\ell^H(r)
    & \leq A_{s,\gamma} \left( \ell+\tfrac 12 \right)^{-4s}\left[\left(\frac{r}{\ell+\tfrac12}\right)^{2s-1}\one_{\{r\leq\ell+\frac12\}}+\left(\frac{r}{\ell+\tfrac12}\right)^{4s-1}\one_{\{\ell+\frac12< r\leq(\ell+\frac12)^2\}}\right.\\
    & \qquad\qquad\qquad\qquad\quad \left.+ \left(\ell+\tfrac12\right)^{4s-1}\one_{\{r>(\ell+\frac 12)^2\}}\right].
  \end{align*}
Moreover, for any $\epsilon>0$ and $r\in\R_+$,
  \begin{align*}
    \rho^H(r) \leq 
\begin{cases}
A_\gamma \, r^{-3/2} & \text{if}\ 0<\gamma<(1+\sqrt 2)/4 \,, \\
A_{\gamma,\epsilon} \left( r^{-2\sigma_\gamma-\epsilon} \one_{\{r\leq 1\}} + r^{-3/2} \one_{\{r>1\}} \right)
& \text{if}\ (1+\sqrt 2)/4\leq \gamma<2/\pi \,.
\end{cases}
  \end{align*}
\end{Thm}

In this theorem, and similarly in the rest of this paper, we denote by $A$ a constant which only depends on the parameters appearing as subscripts of $A$.

We believe that the bound $r^{-3/2}$ on $\rho^H(r)$ for large $r$ is best possible, since this regime should be dominated by non-relativistic effects and in the Schr\"odinger case Heilmann and Lieb \cite{HeilmannLieb1995} showed that the hydrogenic density behaves like $(\sqrt{2}/(3\pi^2)) \gamma^{3/2} r^{-3/2} + o(r^{-3/2})$ for large $r$. Note that this indicates that there is a smooth transition between the
quantum length scale $Z^{-1}$ and the semiclassical length scale $Z^{-1/3}$,
as the Thomas--Fermi density diverges like $(Z/r)^{3/2}$ at the origin. As we will explain in Appendix \ref{a:nonrelrhoh}, arguments similar to those in the proof of Theorem \ref{existencerhoh} also yield an $r^{-3/2}$ bound in the non-relativistic case. Proving a corresponding lower bound on $\rho^H(r)$ for large $r$, as well as determining the asymptotic coefficient is an open problem.

We believe that the bound on $\rho^H$ for small $r$ is best possible for $\gamma\geq (1+\sqrt 2)/4$, except possibly for the arbitrary small $\epsilon>0$. In fact, we believe that $|\psi_{0,0}^H|^2$ has an $r^{2-2\sigma_\gamma}$ behavior at $r=0$. Evidence for this conjecture comes from the trial functions used in \cite[Eq.~(6)]{Raynaletal1994} and from the description of the domain of $C_\ell^H$ in momentum space \cite[Sec.~V]{LeYaouancetal1997}. On the other hand, the small $r$ behavior of the bound in Theorem \ref{existencerhoh} is not optimal for $\gamma<(1+\sqrt 2)/4$ and, indeed, in Theorem \ref{densityagain} we show how to improve it somewhat at the expense of a more complicated statement. The appearance of $\gamma=(1+\sqrt 2)/4$ is technical and comes from the restriction $\sigma\leq 3/4$ together with the fact that $\sigma_{(1+\sqrt 2)/4}=3/4$ (since $\Phi(3/4)=(1+\sqrt 2)/4$). It is an open problem to decide whether the bound $r^{-2\sigma_\gamma}$ for small $r$ is best possible for all $0<\gamma\leq 2/\pi$. Note that in the Coulomb--Dirac model, where the eigenfunctions are known explicitly, the singularity depends on the coupling constant $\gamma$. In contrast, in the non-relativistic case, Heilmann and Lieb \cite{HeilmannLieb1995} have shown that the hydrogenic density is finite at the origin. They have also shown that the density is monotone decreasing in $r$, which again is an open question in the Chandrasekhar case.

%%%%%%%%%

\subsection{Strategy of the proof}\label{sec:strategy}

Theorem \ref{convfixedl} is proved via a linear response argument and
follows the lines of Lieb and Simon \cite{LiebSimon1977}, Baumgartner
\cite{Baumgartner1976} (using Griffiths' lemma \cite{Griffiths1964}, see
also \cite[Theorem 1.27]{Simon2011}) and Iantchenko et al
\cite{Iantchenkoetal1996}: 
we differentiate with respect to the coupling constant $\lambda$ of a
perturbation $U$ in the $\ell$-th angular momentum channel of the
Chandrasekhar operator. The derivative, i.e., the response, is given by
the ground state density integrated against $U$.
To prove Theorem \ref{convalll} one needs to justify that one can interchange the sum over $\ell\in\N_0$ with the limit $Z\to\infty$.

In more detail, the proof of our two main results, Theorems \ref{convfixedl} and~\ref{convalll}, relies on four propositions. The first two reduce the multi-particle problem to a one-body problem. We recall that the Chandrasekhar hydrogen operator $C_\ell^H$ in angular momentum channel $\ell$ was introduced in \eqref{eq:defclh}. Moreover, for a self-adjoint operator $A$ we write $A_-=-A\chi_{(-\infty,0)}(A)$.

\begin{Prop}\label{redux1}
Let $\gamma\in(0,2/\pi)$, $\ell_0\in\N_0$ and assume that $U\geq 0$ is a measurable function on $(0,\infty)$ that is form bounded with respect to $C_0^H$. Then, provided $|\lambda|$ is sufficiently small,
$$
\limsup_{Z\to\infty} \int_0^\infty \!\!c^{-3} \rho_{\ell_0,d}(c^{-1}r) U(r)\,\dr
\leq \lambda^{-1} \left( \tr\left( C_{\ell_0}^H -\lambda U \right)_- - \tr \left( C_{\ell_0}^H\right)_- \right)
\ \ \text{if}\ \lambda>0
$$
and
$$
\liminf_{Z\to\infty} \int_0^\infty \!\!c^{-3} \rho_{\ell_0,d}(c^{-1}r) U(r)\,\dr
\geq \lambda^{-1} \left( \tr\left( C_{\ell_0}^H -\lambda U \right)_- - \tr \left( C_{\ell_0}^H\right)_- \right)
\ \ \text{if}\ \lambda<0.
$$
\end{Prop}

\begin{Prop}\label{redux2}
Let $\gamma\in(0,2/\pi)$, $\ell_0\in\N_0$ and assume that $U\geq 0$ is a measurable function on $(0,\infty)$ that is form bounded with respect to $C_0^H$. Assume that there is a sequence $(a_\ell)_{\ell\geq\ell_0}$ and a $\lambda_0>0$ such that for all $0<\lambda\leq \lambda_0$ and all functions $\chi$ on $(0,\infty)$ with $0\leq\chi\leq\gamma/r$,
$$
\tr\left( C_{\ell}^H +\chi -\lambda U \right)_- - \tr \left( C_{\ell}^H+ \chi \right)_-  \leq \lambda a_\ell \,.
$$
Then
$$
\limsup_{Z\to\infty} \sum_{\ell=\ell_0}^\infty (2\ell+1) \int_0^\infty \!\!c^{-3} \rho_{\ell,d}(c^{-1}r) U(r)\,\dr
\leq \sum_{\ell=\ell_0}^\infty (2\ell+1) a_\ell \,.
$$
\end{Prop}

The proof of these two propositions uses rather standard tools and is given in Section~\ref{sec:redux}. These propositions reduce the proof of our main results to the question of differentiability of the functions $\lambda\mapsto \tr\left( C_{\ell_0}^H -\lambda U \right)_-$ at $\lambda=0$ and to uniform bounds on the corresponding difference quotients. We summarize these results in the following two propositions, which are proved in Sections \ref{s:singlel} and \ref{s:alll}. These results constitute the main technical contribution of this paper.

\begin{Prop}\label{diff1}
Let $\gamma\in(0,2/\pi)$, $\ell_0\in\N_0$ and let $U$ be a non-negative function on $(0,\infty)$ such that either $U\in r^{-1}L^\infty_{\mathrm c}([0,\infty))$ or $U\in\cd_\gamma^{(0)}$. Then $\lambda\mapsto \tr\left( C_{\ell_0}^H -\lambda U \right)_-$ is differentiable at $\lambda=0$ with derivative $\int_0^\infty \rho_{\ell_0}^H(r)U(r)\,\dr$.
\end{Prop}

\begin{Prop}
  \label{diff2}
  Let $\gamma\in(0,2/\pi)$ and let $U=U_1+U_2$ be a function on $(0,\infty)$ with non-negative $U_1\in r^{-1}L^\infty_{\mathrm c}([0,\infty))$ and non-negative $U_2\in\cd$. Then there are $\ell_*\in\N_0$, $\lambda_0>0$, $\epsilon>0$ and $A<\infty$ such that for all $\ell\geq\ell_*$, $0<\lambda\leq\lambda_0$ and all functions $0\leq\chi\leq\gamma/r$ on $(0,\infty)$,
  $$
  \left( \tr\left( C_{\ell}^H +\chi -\lambda U \right)_- - \tr \left( C_{\ell}^H + \chi\right)_- \right) \leq A\,\lambda\, (\ell+1/2)^{-2-\epsilon} \,.
  $$
\end{Prop}

Our main results, Theorems \ref{convfixedl} and \ref{convalll}, follow in a routine way from these four propositions. We include the details for the sake of concreteness.

\begin{proof}[Proof of Theorem \ref{convfixedl}]
Since the assertion of the theorem is additive with respect to $U$ and since the positive and negative parts of $U_1$ and $U_2$ again belong to $r^{-1}L^\infty_{\mathrm c}$ and $\cd_\gamma^{(0)}$, respectively, we may assume from now on that $U$ is non-negative and that it belongs either to
$r^{-1} L^\infty_{\mathrm c}$ or to $\cd_\gamma^{(0)}$. Then Proposition \ref{diff1} implies that
$$
\lim_{\lambda\to 0} \lambda^{-1} \left( \tr\left( C_{\ell_0}^H -\lambda U \right)_- - \tr \left( C_{\ell_0}^H\right)_- \right) = \int_0^\infty \rho_{\ell_0}^H(r)U(r)\,\dr \,.
$$
Since $U$ is form bounded with respect to $C_0^H$ (under the $r^{-1} L^\infty_{\mathrm c}$ assumption this follows from Kato's inequality and under the $\cd_\gamma^{(0)}$ assumption it is shown in the proof of Proposition \ref{diff1}), the assertion in Theorem \ref{convfixedl} follows from Proposition~\ref{redux1}.
\end{proof}

\begin{proof}[Proof of Theorem \ref{convalll}]
As in the proof of Theorem \ref{convfixedl} we may assume that $U$ is non-negative. Then Proposition \ref{diff2} implies that the assumption of Proposition \ref{redux2} is satisfied with $a_\ell = A (\ell+1/2)^{-2-\epsilon}$ for some constant $A$ and some $\epsilon>0$. Therefore Proposition \ref{redux2} implies that 
$$
\limsup_{Z\to\infty} \sum_{\ell=\ell_0}^\infty (2\ell+1) \int_0^\infty \!\!c^{-3} \rho_{\ell,d}(c^{-1}r) U(r)\,\dr
\leq A' (\ell_0+1/2)^{-\epsilon} \,.
$$
In particular, the left side is finite. Moreover, by Theorem \ref{convfixedl} and Fatou's lemma,
\begin{align*}
\sum_{\ell=\ell_0}^\infty (2\ell+1) \int_0^\infty \rho_{\ell}^H(r)U(r)\,\dr
& \leq \liminf_{Z\to\infty} \sum_{\ell=\ell_0}^\infty (2\ell+1) \int_0^\infty \!\!c^{-3} \rho_{\ell,d}(c^{-1}r) U(r)\,\dr \\
& \leq A' (\ell_0+1/2)^{-\epsilon} \,.
\end{align*}
We recall \eqref{eq:totalaverage} and bound
\begin{align*}
& \left| \int_{\R^3} c^{-3}\rho_{d}(c^{-1}|x|)U(|x|)\,\dx - \int_{\R^3} \rho^H(|x|) U(|x|)\,\dx \right| \\
& \qquad \leq \sum_{\ell=0}^{\ell_0-1} (2\ell+1) \left| \int_0^\infty \!\!c^{-3} \rho_{\ell,d}(c^{-1}r) U(r)\,\dr - \int_0^\infty \rho_{\ell}^H(r)U(r)\,\dr \right| \\
& \qquad \quad + \sum_{\ell=\ell_0}^{\infty} (2\ell+1) \left( \int_0^\infty \!\!c^{-3} \rho_{\ell,d}(c^{-1}r) U(r)\,\dr + \int_0^\infty \rho_{\ell}^H(r)U(r)\,\dr \right).
\end{align*}
Thus, by Theorem \ref{convfixedl} and the above bounds,
$$
\limsup_{Z\to\infty} \left| \int_{\R^3} c^{-3}\rho_{d}(c^{-1}|x|)U(|x|)\,\dx - \int_{\R^3} \rho^H(|x|) U(|x|)\,\dx \right| \leq 2A' (\ell_0+1/2)^{-\epsilon} \,.
$$
Since $\ell_0$ can be chosen arbitrarily large, we obtain the claimed convergence.
\end{proof}

Let us discuss some of the difficulties that we overcome in this paper. Both quantities $\int_0^\infty c^{-3}\rho_{\ell_0,d}(c^{-1}r)U(r)\,\dr$ and $\int_0^\infty \rho_{\ell_0}^H(r)U(r)\,\dr$ that appear in Theorem~\ref{convfixedl} can be informally thought of as derivatives at $\lambda=0$ of certain energies with a test potential $\lambda U$. For the `multi-particle quantity' $\int_0^\infty c^{-3}\rho_{\ell_0,d}(c^{-1}r)U(r)\,\dr$ one does not actually have to compute a derivative and, in particular, one does not have to worry about interchanging differentiation with the limit $Z\to\infty$. Instead, one can work with difference quotients because of a convexity argument that is behind the proof of Proposition \ref{redux1}. On the other hand, in order to obtain the `one-body quantity' $\int_0^\infty \rho_{\ell_0}^H(r)U(r)\,\dr$ one has to justify differentiability of $\tr (C_{\ell_0}^H-\lambda U)_-$ at $\lambda=0$, as stated in Proposition \ref{diff1}.

The abstract question of differentiability of $\tr(A-\lambda B)_-$ is answered in Theorem~\ref{diff} under the assumption that $(A+M)^{-1/2} B (A+M)^{-1/2}$ is trace class for $M>-\inf\spec A$. This assumption is satisfied, for instance, in the non-relativistic case and therefore leads to an alternative proof of parts of the results in \cite{Iantchenkoetal1996}. In the relativistic case, however, the form trace class condition is not satisfied and one needs a generalization which, besides some technical conditions, requires that $(A+M)^{-s} B (A+M)^{-s}$ is trace class only for some $s>1/2$. However, the gain from allowing $s>1/2$ comes at the expense of working outside of the natural energy space and leads to several complications.

To be more specific, in our application we have $A=C_{\ell_0}^H$ and $B=U$. Replacing for the moment  $C_{\ell_0}^H$ by $C_{\ell_0}$, the operator $(C_{\ell_0}+M)^{-s} U (C_{\ell_0}+M)^{-s}$ is not trace class for $s=1/2$, no matter how nice $U\not\equiv 0$ is. This follows essentially from the fact that $(\sqrt{k^2+1} - 1 +M)^{-1} \not\in L^1(\R_+,\dk)$. On the other hand, for $s>1/2$ one can show that this operator is trace class for a rather large class of functions $U$, see Proposition \ref{genReltrclassnod}. This leaves us with the problem of replacing $C_{\ell_0}^H$ by $C_{\ell_0}$, which is to say showing boundedness of $(C_{\ell_0}+M)^{s}(C_{\ell_0}^H+M)^{-s}$ for some $s>1/2$. When $\ell_0\geq 1$ or when $\ell_0=0$ and $\gamma<1/2$ one can deduce this boundedness for $s=1$ from Hardy's inequality and obtain the corresponding result for all $s<1$ by operator monotonicity, see Remark \ref{elementary}. In order to deal with the remaining case $\ell_0=0$ and $1/2\leq\gamma<2/\pi$ we need the recent result from \cite{Franketal2019} which says that $|p|^s (|p|-\gamma |x|^{-1})^{-s}$ is bounded in $L^2(\R^3)$ if $s<3/2-\sigma_\gamma$ with $\sigma_\gamma$ from \eqref{eq:defsigma}. Since $\sigma_\gamma<1$ for $\gamma<2/\pi$ we can therefore find an $s>1/2$ such that $(C_{\ell_0}+M)^{s}(C_{\ell_0}^H+M)^{-s}$ is bounded, see Proposition \ref{hardydomcor}. It is at this point that the assumption $\gamma<2/\pi$ enters.

\medskip

% Organization of the paper
The organization of this paper is as follows: In the next section, we prove Propositions \ref{redux1} and \ref{redux2}. In Section \ref{s:diffsum} we compute in an abstract setting the two-sided derivatives of the sum of the negative eigenvalues of an operator $A-\lambda B$ with respect to $\lambda$. In Section \ref{s:singlel} we show that the conditions of the previous section are fulfilled for a certain class of test potentials and thereby prove Proposition \ref{diff1}.
In Section \ref{s:alll}, we control the $\ell$-dependence of difference quotients for a certain class of test potentials, which leads to Proposition \ref{diff2}. Moreover, using a similar argument we will prove Theorem \ref{existencerhoh}.

\medskip

Yasha Sinai is remarkable not only for the depth of his contributions to probability theory and to mathematical physics but to their breadth.  We hope he enjoys this birthday bouquet.

%%%%%%%%%%%%%%%%%%%

%%%%%%%%%%%%%%%%

\section{Reduction to a one-particle problem}\label{sec:redux}

Our goal in this section is to prove Propositions \ref{redux1} and \ref{redux2} which allow us to pass from a multi-particle problem to a one-body problem.

%%%%%%%%%%%%%%%%

\subsection{Proof of Proposition \ref{redux1}}
\label{ss:prooffixedl}

Let
$$
\Pi_{\ell} := \sum_{m=-\ell}^{\ell} |Y_{\ell,m}\rangle\langle Y_{\ell,m}|
$$
be the orthogonal projection onto the subspace of angular momentum $\ell$ and define the operator
$$
C_{Z,\lambda}:=C_Z-\lambda\sum_{\nu=1}^N c^2 U(c |x_\nu|)\Pi_{\ell_0,\nu}
\quad\text{in}\ \bigwedge_{\nu=1}^N L^2(\R^3) \,.
$$
Here the operator $\Pi_{\ell_0,\nu}$ acts as $\Pi_{\ell_0}$ with
respect to the $\nu$-th particle. Since $U$ is assumed to be form bounded with respect to $C_\ell^H$, the operator $C_{Z,\lambda}$ can be defined in the
sense of quadratic forms for all $\lambda$ in an open neighborhood of zero,
which is independent of $Z$.

The starting point of the proof is that the quantity of interest can be written as
\begin{equation}
  \label{eq:5.2}
  \int_0^\infty c^{-3} \rho_{\ell_0,d}(c^{-1}r)U(r)\,\dr = \lambda^{-1} (2\ell_0+1)^{-1} c^{-2}  \ttr (C_Z-C_{Z,\lambda})d \,.
\end{equation}

In order to prove Proposition \ref{redux1} we will bound $\ttr C_{Z,\lambda}d$ from below and $\ttr C_Z d$ from above.

We begin with the lower bound on $\ttr(d C_{Z,\lambda})$, which we will obtain through a correlation inequality. We denote by $\rho_Z^{\rm TF}$ the unique minimizer of the Thomas--Fermi
functional for a neutral atom with ground state energy $E^{\rm TF}(Z)$
(Lieb and Simon \cite[Theorem II.20]{LiebSimon1977}).
Moreover, we define the radius $R_Z^{\mathrm{TF}}(x)$ of the exchange hole at
$x\in\R^3$ by
$$
\int\limits_{|x-y| \leq R_Z^{\mathrm{TF}}(x)}\rho_Z^{\mathrm{TF}}(y)\,\dy=\frac12,
$$
set
$$
\chi_Z^{\mathrm{TF}}(x) := \int\limits_{|x-y| \geq R_Z^{\mathrm{TF}}(x)}\frac{\rho_Z^{\mathrm{TF}}(y)}{|x-y|}\,\dy
$$
and recall the correlation inequality by Mancas et al
\cite{Mancasetal2004},
\begin{align}
  \label{eq:correlation}
  \sum_{\nu<\mu}\frac{1}{|x_\nu-x_\mu|}\geq\sum_{\nu=1}^Z\chi_Z^{\mathrm{TF}}(x_\nu)-D[\rho_Z^{\mathrm{TF}}]\,.
\end{align}

For a self-adjoint operator $v$ in $L^2(\R^3)$ which is form bounded with respect to $\sqrt{-\Delta}$ with form
bound $<c$ we define
$$
C_c(v) = \sqrt{-c^2\Delta + c^4} - c^2 - v \quad\text{in}\ L^2(\R^3) \,. 
$$
Moreover, for a trace class operator $A$ in $L^2(\R^3)$, we define
$$
\ttr_\ell A := \ttr \,\Pi_\ell A\Pi_\ell \,.
$$
We now bound $\tr C_{Z,\lambda}d$ from below in terms of traces of one-particle
operators.

\begin{Lem}
  \label{correlation}
  For all $\lambda$ in a neighborhood of $0$ and all $\N\ni L<Z$,
  \begin{align*}
    \tr C_{Z,\lambda} d & \geq - \sum_{\ell=0}^{L-1} \ttr_\ell C_c(Z|x|^{-1} +\lambda c^2 U(c|x|)\Pi_{\ell_0})_- \\
                        & \quad - \sum_{\ell=L}^Z \ttr_\ell C_c(Z|x|^{-1} - \chi_Z^{\rm TF} +\lambda c^2 U(c|x|)\Pi_{\ell_0})_- - D[\rho_Z^{\rm TF}] \,.
  \end{align*}
\end{Lem}

\begin{proof}
  Let $d^{(1)}$ denote the one-particle density matrix of $d$.
  Applying the correlation inequality \eqref{eq:correlation} and
  using the non-negativity and spherical symmetry of
  $\chi_Z^{\mathrm{TF}}(x)$, we obtain for any $L<Z$,
  \begin{align*}
    \tr C_{Z,\lambda}d & \geq \sum_{\ell=0}^{\infty} \tr_\ell C_c(Z|x|^{-1} -\chi_Z^{\mathrm{TF}} +\lambda c^2 U(c|x|)\Pi_{\ell_0})d^{(1)} - D[\rho_Z^{\mathrm{TF}}] \\
& \geq \sum_{\ell=0}^{L-1} \tr_\ell C_c(Z|x|^{-1} +\lambda c^2 U(c|x|)\Pi_{\ell_0})d^{(1)} \\
    & \quad + \sum_{\ell=L}^\infty \tr_\ell C_c(Z|x|^{-1} -\chi_Z^{\mathrm{TF}} +\lambda c^2 U(c|x|)\Pi_{\ell_0})d^{(1)}  - D[\rho_Z^{\mathrm{TF}}] \,.
  \end{align*}
Since the restriction of the operator $C_c(Z|x|^{-1} -\chi_Z^{\mathrm{TF}} +\lambda c^2 U(c|x|)\Pi_{\ell_0})$ to angular momentum $\ell$ is increasing in $\ell$, we can estimate
  the last expression further from below by replacing $d^{(1)}$ by a
  one-particle density matrix that is defined such that all channels
  $\ell<L$ are completely occupied. Since there are no more than $Z$
  total angular momentum channels occupied anyway, the second sum can
  be cut off at $Z$. Finally, invoking the variational principle yields
  the claimed bound.
\end{proof}

Our next goal is to estimate $\tr C_{Z}d=\inf\spec C_Z$ from above using the results from \cite{Franketal2008}.
  
\begin{Lem}
  \label{energyasymp}
  If $L=[Z^{1/9}]$, then
  \begin{align}
    \label{eq:lowerbound}
    \inf\spec C_Z & \leq - \sum_{\ell=0}^{L-1} \ttr_\ell C_c(Z|x|^{-1})_- - \sum_{\ell=L}^Z \ttr_\ell C_c(Z|x|^{-1}- \chi_Z^{\rm TF})_-  \notag \\
    & \quad - D[\rho_Z^{\rm TF}]  + \const Z^{47/24}
      \,.
  \end{align}
\end{Lem}

\begin{proof}
We denote by $S_Z$ the non-relativistic analogue of $C_Z$ which is given by the same formula but with $\sqrt{-c^2\Delta+c^4}-c^2$ replaced by $-(1/2)\Delta$. Similarly, we denote by $S(v)$ the non-relativistic analogue of $C_c(v)$ and set
\begin{align*}
\Delta^C(Z)\! & := \inf\spec C_Z \!+\! \sum_{\ell=0}^{L-1} \ttr_\ell C_c(Z|x|^{-1})_- \!+\! \sum_{\ell=L}^Z \ttr_\ell C_c(Z|x|^{-1}- \chi_Z^{\rm TF})_-  \! +\! D[\rho_Z^{\rm TF}], \\
\Delta^S(Z)\! & := \inf\spec S_Z + \sum_{\ell=0}^{L-1} \ttr_\ell S(Z|x|^{-1})_- + \sum_{\ell=L}^Z \ttr_\ell S(Z|x|^{-1}- \chi_Z^{\rm TF})_-  + D[\rho_Z^{\rm TF}].
\end{align*}
Then
\begin{equation}
\label{eq:energyproof1}
\inf\spec S_Z - \inf\spec C_Z = \Delta^S(Z)-\Delta^C(Z) + Z^2 s(\gamma) + \mathcal O(Z^{17/9})
\end{equation}
for a certain constant $s(\gamma)$. The analogue of this bound in the Brown--Ravenhall case is proved in \cite[Subsection 4.1, Proof of Theorem 1.1 -- First part]{Franketal2009}, but extends to the Chandrasekhar case; see also the slightly less precise version in \cite[Proof of Theorem 1 -- First part]{Franketal2008}. On the other hand, we have
\begin{equation}
\label{eq:energyproof2}
\inf\spec S_Z - \inf\spec C_Z \geq Z^2 s(\gamma) - \const Z^{47/24} \,.
\end{equation}
Again, in the Brown--Ravenhall case this is proved in \cite[Subsection 4.2.3]{Franketal2009}, but it extends, with a simpler proof, to the Chandrasekhar case. (We note that the corresponding bound in \cite[Proof of Theorem 1 -- Second part]{Franketal2008} only gives a $o(Z^2)$ error.) Combining \eqref{eq:energyproof1} and \eqref{eq:energyproof2} we obtain
$$
\Delta^S(Z)-\Delta^C(Z) \geq -\const Z^{47/24} \,.
$$
Since $\Delta^S(Z) = \mathcal O(Z^{47/24})$ (which is, essentially, \cite[Proposition 4.1]{Franketal2009}, which is similar to \cite[Proposition 3]{Franketal2008}), we deduce that $\Delta^C(Z) \leq \const Z^{47/24}$, as claimed in the lemma.
\end{proof}

After these preliminaries we begin with the main part of the proof of Proposition~\ref{redux1}. Inserting the bounds from Lemmas \ref{correlation} (with $L=[Z^{1/9}]$) and \ref{energyasymp} into \eqref{eq:5.2}, we obtain
\begin{align*}
& c^2 \lambda (2\ell_0+1) \int_0^\infty c^{-3} \rho_{\ell_0,d}(c^{-1}r)U(r)\,\dr =  \ttr (C_Z-C_{Z,\lambda})d \\
& \quad \leq 
  \sum_{\ell=0}^{L-1} \left(\ttr_\ell C_c(Z|x|^{-1}+\lambda c^2 U(c|x|)\Pi_{\ell_0})_- - \ttr_\ell C_c(Z|x|^{-1})_- \right) \\
  & \quad \quad + \sum_{\ell=L}^Z \left(\ttr_\ell C_c(Z|x|^{-1}+\lambda c^2 U(c|x|)\Pi_{\ell_0} - \chi_Z^{\rm TF})_-  - \ttr_\ell C_c(Z|x|^{-1}-\chi_Z^{\rm TF})_- \right) \\
  & \quad\quad + \const Z^{47/24} \,.
\end{align*}
For sufficiently large $Z$, we have $L=[Z^{1/9}]>\ell_0$. Thus, the
last expression simplifies to
\begin{align*}
& c^2 \lambda (2\ell_0+1) \int_0^\infty c^{-3} \rho_{\ell_0,d}(c^{-1}r)U(r)\,\dr \\
& \quad \leq 
  \ttr_{\ell_0} C_c(Z|x|^{-1}+\lambda c^2 U(c|x|))_- - \ttr_\ell C_c(Z|x|^{-1})_- + \const Z^{47/24} \\
& \quad = c^2 \left( \ttr_{\ell_0} C_1(\gamma|x|^{-1}+\lambda U(|x|))_- - \ttr_\ell C_1(\gamma |x|^{-1})_- + \const Z^{-1/24} \right)  \\ 
& \quad = c^2 (2\ell_0+1) \left( \tr\left( C_{\ell_0}^H -\lambda U \right)_- - \tr \left( C_{\ell_0}^H\right)_- + \const Z^{-1/24} \right).
\end{align*}
Letting $Z\to\infty$ we obtain
$$
\limsup_{Z\to\infty} \lambda \int_0^\infty c^{-3} \rho_{\ell_0,d}(c^{-1}r) U(r)\,\dr
\leq \tr\left( C_{\ell_0}^H -\lambda U \right)_- - \tr \left( C_{\ell_0}^H\right)_-  \,.
$$
This implies the bounds in the proposition.
\qed

%%%%%%%%%%%%%%%%%%%%%%%

\subsection{Proof of Proposition \ref{redux2}}

Similarly as in the previous subsection, for $\ell_0\in\N$ we introduce
\begin{align*}
  C_{Z,\lambda}^{\ell_0} & := C_Z - \lambda \sum_{\nu=1}^N c^2U(c|x_\nu|)\sum_{\ell=\ell_0}^{\infty}\Pi_{\ell,\nu} \qquad \text{in}\ \bigwedge_{\nu=1}^N L^2(\R^3)\,.
\end{align*}
As in \eqref{eq:5.2}, we have
\begin{align}
  \label{eq:weakdensity}
  \sum_{\ell=\ell_0}^{\infty}(2\ell+1)\int_0^\infty c^{-3}\rho_{\ell,d}(r/c)U(r)\,\dr
  = \lambda^{-1} c^{-2}\tr (C_Z-C_{Z,\lambda}^{\ell_0})d \,.
\end{align}
Note that both sides are well-defined although possibly equal to $+\infty$. The left side is a sum of non-negative terms and on the right side, $\tr C_Z d=\inf\spec C_Z>-\infty$.

Combining identity \eqref{eq:weakdensity} with an obvious generalization of Lemma \ref{correlation} and with Lemma \ref{energyasymp} we obtain for $L>\ell_0$,
\begin{align*}
& c^2 \lambda \sum_{\ell=\ell_0}^{\infty}(2\ell+1)\int_0^\infty c^{-3}\rho_{\ell,d}(r/c)U(r)\,\dr \\
& \leq 
  \sum_{\ell=0}^{L-1} \left(\ttr_\ell C_c(Z|x|^{-1}+\lambda c^2 U(c|x|)\! \sum_{\ell'=\ell_0}^{\infty} \!\Pi_{\ell'})_- - \ttr_\ell C_c(Z|x|^{-1})_- \right) \\
  & \quad + \sum_{\ell=L}^Z \left(\ttr_\ell C_c(Z|x|^{-1}+\lambda c^2 U(c|x|)\! \sum_{\ell'=\ell_0}^{\infty} \!\Pi_{\ell'} - \chi_Z^{\rm TF})_-  - \ttr_\ell C_c(Z|x|^{-1}-\chi_Z^{\rm TF})_- \right) \\
  & \quad + \const Z^{47/24} \\
  & =
  c^2 \sum_{\ell=\ell_0}^{L-1} \left(\ttr_\ell C_1(Z|x|^{-1}+\lambda c^2 U(c|x|))_- - \ttr_\ell C_1(Z|x|^{-1})_- \right)
   \\
  & \quad + c^2 \sum_{\ell=L}^Z \left(\ttr_\ell C_1(\gamma |x|^{-1}+\lambda U(|x|) - c^{-2} \chi_Z^{\rm TF}(x/c))_-  \right. \\
  & \qquad\qquad\quad \left. - \ttr_\ell C_1(\gamma |x|^{-1}-c^{-2} \chi_Z^{\rm TF}(x/c))_- \right)   
  \\
  & \quad + \const Z^{47/24} \,.
\end{align*}
Since $c^{-2} \chi_Z^{\rm TF}(x/c)$ is radial (because the Thomas--Fermi density is radial) and $0\leq c^{-2} \chi_Z^{\rm TF}(x/c) \leq \gamma/r$ (the second inequality here follows from the fact that the Thomas--Fermi potential is non-negative), the assumption of the proposition implies that for $0<\lambda\leq\lambda_0$
\begin{align*}
\sum_{\ell=\ell_0}^{\infty}(2\ell+1)\int_0^\infty c^{-3}\rho_{\ell,d}(r/c)U(r)\,\dr \leq
  \sum_{\ell=\ell_0}^{\infty} (2\ell+1) a_\ell + \const \lambda^{-1} Z^{-1/24} \,.
\end{align*}
Taking the limsup as $Z\to\infty$, we obtain the bound in the proposition.
\qed

\begin{Rem}\label{approxgs1}
In Remark \ref{approxgs} we claim that that Theorems \ref{convfixedl} and \ref{convalll} continue to hold for approximate ground states in the sense of \eqref{eq:approxgs}. To justify this claim let us show that Propositions \ref{redux1} and \ref{redux2} continue to hold in this more general set-up. In fact, Lemma \ref{energyasymp} and \eqref{eq:approxgs} now imply
  \begin{align*}
     \tr C_Z d & \leq - \sum_{\ell=0}^{L-1} \ttr_\ell C_c(Z|x|^{-1})_- - \sum_{\ell=L}^Z \ttr_\ell C_c(Z|x|^{-1}- \chi_Z^{\rm TF})_-  \notag \\
    & \quad - D[\rho_Z^{\rm TF}]  + o(Z^2) \,.
  \end{align*}
The rest of the proof remains unchanged. With the analogues of Propositions \ref{redux1} and \ref{redux2} for approximate ground states in place, the analogues of Theorems \ref{convfixedl} and \ref{convalll} follow by the same arguments as in Subsection \ref{sec:strategy}.
\end{Rem}

%%%%%%%%%%%%%%%%%%%

\section{Differentiability of the sum of negative eigenvalues}
\label{s:diffsum}

\subsection{Differentiating under a relative trace class assumption}

We say that an operator $B$ is \emph{relatively form trace class} with respect to a lower bounded self-adjoint operator $A$ if $(A+M)^{-1/2}B(A+M)^{-1/2}$ is trace class for some (and hence any) large enough $M>0$. We recall that we use the notation $A_-=-A\chi_{(-\infty,0)}(A)$.

\begin{Thm}
  \label{diff}
  Assume that $A$ is self-adjoint with $A_-$ trace class.
  Assume that $B$ is non-negative and relatively form trace class with
  respect to $A$. Then the one-sided derivatives of
  $$
  \lambda\mapsto S(\lambda) := \Tr(A-\lambda B)_-
  $$
  satisfy
\begin{equation}
\label{eq:diff}
  \Tr B \chi_{(-\infty,0)}(A) = D^-S(0) 
  \leq D^+S(0) = \Tr B \chi_{(-\infty,0]}(A) \,.
\end{equation}
  In particular, $S$ is differentiable at $\lambda=0$ if and only
  if $B|_{\ker A} = 0$.
\end{Thm}

\emph{Remarks.} (1) Note that the relative trace class assumption implies
that the expression on the right of \eqref{eq:diff}, and consequently also that on the left, is finite. In fact, denoting $P=\chi_{(-\infty,0]}(A)$, we find
$$
\Tr PB = \Tr \Big( P(A+M) \Big) \Big( (A+M)^{-1/2}B(A+M)^{-1/2} \Big) <\infty \,,
$$
since $P(A+M)$ is bounded.\\
(2) It follows from the variational principle that $S$ is convex.
Therefore, by general arguments, $S$ has left and right sided derivatives.\\
(3) If the bottom of the essential spectrum of $A$ is strictly positive,
then the result is well known and will actually be used in our proof. Our
point is that the formulas remain valid even when the bottom of the
essential spectrum is zero, so that perturbation theory is not (directly)
applicable.

\begin{proof}
  \emph{Step 1.} 
  We claim that for any $\lambda\in\R$, $(A-\lambda B)_-$ is trace class
  and that
  $$
  S(\lambda) -S(0) = \int_0^\lambda T(\lambda')\,\td\lambda'
  $$
  with 
  $$
  T(\lambda) := \Tr B \chi_{(-\infty,0)}(A-\lambda B) \,.
  $$
  
  Note that $S(0)$ is finite by assumption. Moreover, $T(\lambda)$ is
  finite for any $\lambda\in\R$, since relative form boundedness of $B$
  implies that $(A-\lambda B+M)^{-1/2} (A+M)^{1/2}$ is bounded and therefore
  $B$ is relatively trace class with respect to $A-\lambda B$, so
  $T(\lambda)<\infty$ follows in the same way as $\tr PB<\infty$ in the
  first remark above. This argument also shows that the integral above
  is finite.

  In order to prove the claimed trace class property and the formula
  for $S(\lambda)$, we let $\mu\in(-\infty,0)\cap\rho(A)$ and set
  $$
  S_\mu(\lambda) := \Tr(A-\lambda B-\mu)_-
  \qquad\text{and}\qquad
  T_\mu(\lambda) := \Tr B \chi_{(-\infty,\mu)}(A-\lambda B) \,.
  $$
  Since $B$ is relatively compact and the infimum of the essential
  spectrum of $A$ is non-negative, $A-\lambda B$ has only finitely many
  eigenvalues below $\mu$. Moreover, by standard perturbation theory, the
  function $\lambda\mapsto S_\mu(\lambda)$ is differentiable at any
  $\lambda$ for which $\mu\not\in\sigma_p(A-\lambda B)$ with derivative
  $T_\mu(\lambda)$. By the Birman--Schwinger principle, the condition
  $\mu\not\in\sigma_p(A-\lambda B)$ is equivalent to
  $1/\lambda\not\in\sigma(B^{1/2}(A-\mu)^{-1} B^{1/2})$, which, since $B$ is
  relatively compact, is true on the complement of a discrete set.
  Therefore, for any $\lambda\in\R$,
  $$
  S_\mu(\lambda) = S_\mu(0) + \int_0^\lambda T_\mu(\lambda')\,\td\lambda' \,.
  $$

  We now let $\mu\to 0-$. Since $\mu\mapsto S_\mu(0)$ and
  $\mu\mapsto T_\mu(\lambda')$ are non-decreasing with finite limit $S(0)$
  and finite, integrable limit $T(\lambda')$, respectively, we conclude
  that the limit $S(\lambda)$ of $S_\mu(\lambda)$ as $\mu\to 0-$ is finite
  and satisfies the required equality.

  \emph{Step 2.} We claim that
  $$
  \limsup_{\lambda\to 0+} T(\lambda) \leq \Tr B \chi_{(-\infty,0]}(A)
  $$
  and
  $$
  \liminf_{\lambda\to 0-} T(\lambda) \geq \Tr B \chi_{(-\infty,0)}(A) \,.
  $$
  This, together with Step 1, immediately implies 
  \begin{equation}
    \label{eq:step2}
    \Tr B \chi_{(-\infty,0)}(A) \leq D^-S(0) 
    \leq D^+S(0) \leq \Tr B \chi_{(-\infty,0]}(A) \,.
  \end{equation}
  
  For $\epsilon>0$ let $f_\epsilon^+$ be the
  function which is $1$ on $(-\infty,0]$, $0$ on $[\epsilon,\infty)$ and
  linear in-between. Similarly, let $f_\epsilon^-$ be the function which
  is $1$ on $(-\infty,-\epsilon]$, $0$ on $[0,\infty)$ and linear
  in-between. Thus, $f^-_\epsilon \leq \chi_{(-\infty,0)}\leq f^+_\epsilon$
  and therefore
  \begin{equation}
    \label{eq:regularization}
    \Tr B f_\epsilon^-(A-\lambda B) \leq T(\lambda) \leq \Tr B f_\epsilon^+(A-\lambda B) \,.
  \end{equation}
  We claim that for any $\epsilon>0$
  \begin{equation}
    \label{eq:src}
    \lim_{\lambda\to 0} \Tr B f_\epsilon^\pm(A-\lambda B) = \Tr B f_\epsilon^\pm(A) \,,
  \end{equation}
  and that
  \begin{equation}
    \label{eq:domconv}
    \limsup_{\epsilon\to 0+}\Tr B f_\epsilon^+(A) = \Tr B \chi_{(-\infty,0]}(A) \,,
    \quad
    \liminf_{\epsilon\to 0+}\Tr B f_\epsilon^-(A) = \Tr B \chi_{(-\infty,0)}(A) \,.
  \end{equation}
  Once we have shown these two facts we can first let $\lambda\to0$ and then
  $\epsilon\to 0+$ in \eqref{eq:regularization} and obtain the claim.

  To prove \eqref{eq:src} we write
  $$
  \Tr B f_\epsilon^\pm(A-\lambda B) = \Tr C K(\lambda) g_\epsilon^\pm(A-\lambda B) K(\lambda)^* 
  $$
  with $C = (A+M)^{-1/2} B (A+M)^{-1/2}$,
  $K(\lambda) = (A+M)^{1/2} (A-\lambda B+M)^{-1/2}$ and
  $g_\epsilon^\pm(\alpha) = (\alpha+M) f_\epsilon^\pm(\alpha)$.
  Since $A-\lambda B$ converges in norm resolvent sense to $A$ as
  $\lambda\to 0$ and since $g_\epsilon^\pm$ are continuous, we have
  $g_\epsilon^\pm(A-\lambda B) \to g_\epsilon^\pm (A)$ in norm
  \cite[Theorem VIII.20]{ReedSimon1972}. Moreover, it is easy to see
  that $K(\lambda)^*$ converges strongly to the identity. (On elements
  in $\ran (A+M)^{-1/2}$ this follows from strong resolvent convergence
  of $A-\lambda B$ and for general elements one uses the uniform
  boundedness of $K(\lambda)^*$ with respect to $\lambda$, which follows
  from the boundedness of $B$ relative to $A$.) We conclude that
  $K(\lambda) g_\epsilon^\pm(A-\lambda B) K(\lambda)^*$ converges weakly to $g_\epsilon^\pm(A)$. Since $C$ is
  trace class, this implies \eqref{eq:src}.

  To prove \eqref{eq:domconv} we write similarly
  $$
  \Tr B f_\epsilon^\pm(A)
  = \int_\R g_\epsilon^\pm(\alpha)\, \td \left( \sum_n c_n (\psi_n, E(\alpha)\psi_n) \right),
  $$
  where $C = \sum_n c_n |\psi_n\rangle\langle\psi_n|$ and $\td E$ is the
  spectral measure for $A$. The functions $g_\epsilon^+$ and $g_\epsilon^-$
  are bounded on the support of $\td E$ and converge pointwise to
  $(\alpha+M) \chi_{(-\infty,0]}(\alpha)$ and
  $(\alpha+M) \chi_{(-\infty,0)}(\alpha)$, respectively, as $\epsilon\to 0+$.
  Since $\td \sum_n c_n (\psi_n, E(\alpha)\psi_n)$ is a finite measure,
  dominated convergence implies that
  $$
  \lim_{\epsilon\to 0+} \!\Tr B f_\epsilon^+(A)
  = \!\int_\R (\alpha+M) \chi_{(-\infty,0]} (\alpha)\, \td \!\left( \!\sum_n c_n (\psi_n, E(\alpha)\psi_n) \!\right)\!
  = \Tr B \chi_{(-\infty,0]}(A)
  $$
  and
  $$
  \lim_{\epsilon\to 0+} \!\Tr B f_\epsilon^-(A)
  = \!\int_\R (\alpha+M) \chi_{(-\infty,0)} (\alpha)\, \td \! \left( \!\sum_n c_n (\psi_n, E(\alpha)\psi_n) \! \right)\!
  = \Tr B \chi_{(-\infty,0)}(A).
  $$
  This proves \eqref{eq:domconv}.

  \emph{Step 3.} We prove that the left and the right inequality in \eqref{eq:step2} are, in fact, equalities.

  By the variational principle, the functions $S_\mu$ are convex and
  converge pointwise to $S$ as $\mu\to 0-$. Thus, by general facts about
  convex functions (see, e.g., \cite[Theorem~1.27]{Simon2011}),
  $$
  D^-S(0) \leq \liminf_{\mu\to0-} D^-S_\mu(0) \,.
  $$ 
  It is well known that
  $$
  D^-S_\mu(0) = T_\mu(0) \,.
  $$
  By monotone convergence,
  $$
  \lim_{\mu\to 0-} T_\mu(0) =  \Tr B \chi_{(-\infty,0)}(A) \,,
  $$
  and therefore $D^-S(0) \leq \Tr B \chi_{(-\infty,0)}(A)$. Thus, the left inequality in \eqref{eq:step2} is an equality.
  
  We abbreviate again $P=\chi_{(-\infty,0]}(A)$ and note that by the
  variational principle
  $$
  -\tr(A-\lambda B)_- \leq \Tr(A-\lambda B)P \,.
  $$
  Thus,
  $$
  S(\lambda) -S(0) \geq \lambda \Tr BP
  $$
  and
  $$
  D^+S(0) = \lim_{\lambda\to 0+} \frac{S(\lambda)-S(0)}{\lambda} \geq \Tr BP \,, 
  $$
  which shows that the right inequality in \eqref{eq:step2} is an equality.
\end{proof}

%%%%%%%%%%%%%%%%%%%%%%%

\subsection{A generalization}

In the application that we have in mind the relative trace class assumption in Theorem \ref{diff} is too strong.
In this subsection we present a generalization of Theorem \ref{diff} where
this assumption is replaced by the weaker assumption that $B$ is relatively
form trace class with respect to $(A+M)^{2s}$ for some $s>1/2$. However, in
this situation we also need to require that the operators $(A+M)^{s}$ and
$(A-\lambda B+M)^{s}$ are comparable in a certain sense.

\begin{Thm}
  \label{diffgen0}
  Assume that $A$ is self-adjoint with $A_-$ trace class. Assume that $B$ is
  non-negative and relatively form bounded with respect to $A$. Assume that
  there are $1/2< s\leq 1$ such that for some $M>-\inf\spec A$,
    \begin{equation}
    \label{eq:traceclassdelta0}
    (A+M)^{-s} B(A+M)^{-s} \qquad\text{is trace class}
  \end{equation}
  and
  \begin{equation}
    \label{eq:relbounddelta}
    \limsup_{\lambda\to 0} \left\| (A+M)^{s} (A-\lambda B+M)^{-s} \right\| <\infty \,.
  \end{equation}
  Then the conclusions in Theorem \ref{diff} are valid.
\end{Thm}

Note that, since $B$ is relatively form bounded with respect to $A$, for any $M>-\inf\spec A$ there is a $\lambda_M$ such that $A-\lambda B+M\geq 0$ for all $|\lambda|\leq \lambda_M$. Therefore $(A-\lambda B+M)^{-s}$ is well-defined for $|\lambda|\leq\lambda_M$.

\begin{proof}[Proof of Theorem \ref{diffgen0}]
  We follow the steps in the proof of Theorem \ref{diff}. At the beginning
  of Step 1 we needed to show that
  $T(\lambda)=\Tr B \chi_{(-\infty,0)}(A-\lambda B)$ is finite and
  uniformly bounded for $\lambda$ near zero. This follows from
  \begin{equation}
    \label{eq:traceclassdelta}
    \limsup_{\lambda\to 0} \tr (A-\lambda B+M)^{-s} B (A-\lambda B+M)^{-s} <\infty
  \end{equation}  
  and the fact that $(A-\lambda B+M)^{2s}\chi_{(-\infty,0)}(A-\lambda B)$ is
  uniformly bounded for $\lambda$ near zero. Note that
  \eqref{eq:traceclassdelta} follows from \eqref{eq:traceclassdelta0} and
  \eqref{eq:relbounddelta}.

  Furthermore, in Step 1, we also used the fact that any given
  $\mu\in(-\infty,0)\cap\rho(A)$ is not an eigenvalue of $A-\lambda B$
  away from a discrete set of $\lambda$'s near zero. Let us justify this
  fact under the present assumptions. We first note that
  $$
  (A-\lambda B+M)^{-1} - (A+M)^{-1} = (A-\lambda B+M)^{-1+s} D(\lambda) E(\lambda) (A+M)^{-1+s}
  $$
  with
  $$
  D(\lambda) =
  (A-\lambda B+M)^{-s} (A+M)^{s}
  $$
  and
  $$
  E(\lambda) = \lambda(A+M)^{-s} B (A+M)^{-s} \,.
  $$
  By assumption \eqref{eq:relbounddelta}, $D(\lambda)$ is bounded and,
  by assumption \eqref{eq:traceclassdelta0}, $E(\lambda)$ is trace class.
  Since $s\leq 1$, this shows that $(A-\lambda B+M)^{-1} - (A+M)^{-1}$ is
  trace class and, in particular, compact. Therefore, by Weyl's theorem, the
  negative spectrum of $A-\lambda B$ is discrete. Since $B$ is relatively
  form bounded with respect to $A$, $A-\lambda B$ forms an analytic family
  of type (B) \cite[Chapter Seven, Theorem 4.8]{Kato1966} and therefore,
  locally, the eigenvalues can be labeled to be analytic functions
  of $\lambda$. Since, by assumption $\mu$ is not an eigenvalue of $A$,
  there is only a discrete set of $\lambda$'s near zero such that $\mu$
  is an eigenvalue of $A-\lambda B$, as claimed.

  Turning now to Step 2, we need to show \eqref{eq:src}. We write again
  $$
  \Tr B f_\epsilon^\pm(A-\lambda B) = \Tr C K(\lambda) g_\epsilon^\pm(A-\lambda B) K(\lambda)^* 
  $$
  where now $C = (A+M)^{-s} B (A+M)^{-s}$,
  $K(\lambda) = (A+M)^{s} (A-\lambda B+M)^{-s}$, and
  $g_\epsilon^\pm(\alpha) = (\alpha+M)^{2s} f_\epsilon^\pm(\alpha)$. We again
  have $g_\epsilon^\pm(A-\lambda B)\to g_\epsilon^\pm(A)$ in norm. In order
  to show that $K(\lambda)^*\to 1$ strongly, we observe again that this
  holds on elements in $\ran (A+M)^{-s}$ and that $K(\lambda)^*$ is
  uniformly bounded for $\lambda$ near zero by assumption
  \eqref{eq:relbounddelta}. Thus, as before,
  $K(\lambda) g_\epsilon^\pm(A-\lambda B) K(\lambda)^*\to g_\epsilon^\pm(A)$
  in the sense of weak operator convergence and, since $C$ is trace class
  by assumption \eqref{eq:traceclassdelta0}, we obtain \eqref{eq:src}.

  Finally, Step 3 remains unchanged. This concludes the proof of
  Theorem~\ref{diffgen0}.
\end{proof}

Let us give a sufficient condition for \eqref{eq:relbounddelta}.

\begin{Prop}
  \label{diffgen}
  Assume that $A$ is self-adjoint and $B$ is non-negative and let $1/2< s\leq 1$. Assume that there is an $s'<s$ such that for some $M>0$ and $a>0$, 
  \begin{equation}
    \label{eq:increasedrelbdd}
    B^{2s} \leq a(A+M)^{2s'} \,.
  \end{equation}
Then $B$ is form bounded with respect to $A$ with form bound zero and \eqref{eq:relbounddelta} holds.
\end{Prop}

The assumption $s'<s$ is crucial for our proof, but we do not know
whether it is necessary for \eqref{eq:relbounddelta} to hold.

Proposition \ref{diffgen} is an immediate consequence of the following lemma, where $A$, $B$, $\alpha$, $\beta$ play the roles of $A+M$, $-\lambda B$, $s$ and $s'$, respectively. The statement and proof of this lemma are inspired by Neidhardt and Zagrebnov \cite[Lemma 2.2]{NeidhardtZagrebnov1999}.

\begin{Lem}
  \label{apriori}
  Let $A$ be a self-adjoint operator with $\inf\spec A>0$ and let $B$ be an
  operator which satisfies $B\geq 0$ or $B\leq 0$. Assume that for some
  numbers $\max\{\beta,1/2\}<\alpha<1$ one has
  \begin{equation*}
    \||B|^\alpha A^{-\beta}\|<\infty \,.
  \end{equation*}
  Then $B$ is form bounded with respect to $A$ with relative bound zero and, if $M\geq C \||B|^\alpha A^{-\beta}\|^{1/(\alpha-\beta)}$ for some
  constant $C$ depending only on $\alpha$ and $\beta$,
  $$
  \frac12 (A+M)^{2\alpha} \leq (A+B+M)^{2\alpha} \leq 2(A+M)^{2\alpha} \,.
  $$
\end{Lem}

The constants $1/2$ and $2$ can be replaced by arbitrary constants
$1-\epsilon$ and $1+\epsilon$ with $\epsilon>0$ at the expense of
choosing $C$ depending on $\epsilon>0$.

\begin{proof}[Proof of Lemma \ref{apriori}]
\emph{Step 0.} By assumption, we have
$$
|B|^{2\alpha} \leq  \||B|^\alpha A^{-\beta}\|^2 A^{2\beta}
$$
and therefore by operator monotonicity of $x\mapsto x^{1/(2\alpha)}$, for any $\epsilon>0$,
$$
|B| \leq  \||B|^\alpha A^{-\beta}\|^{1/\alpha} A^{\beta/\alpha} \leq \tfrac\beta\alpha \epsilon A + (1-\tfrac{\beta}{\alpha}) \||B|^\alpha A^{-\beta}\|^{1/(\beta-\alpha)} \epsilon^{-\beta/(\alpha-\beta)}.
$$
This shows that $B$ is form bounded with respect to $A$ with relative bound zero. Thus, the operator $A+B$ is defined in the sense of quadratic forms.

  \emph{Step 1.} We begin by showing that under the additional assumptions
  \begin{equation}
    \label{eq:small}
    \left\| |B|^\alpha A^{-\alpha} \right\|<1 \,,
  \end{equation}
  as well as
  \begin{equation}
    \label{eq:invertible}
    A+B>0 \,,
  \end{equation}
  there is an operator $S$ satisfying
  \begin{equation}
    \label{eq:abs}
    (A+B)^{-\alpha} = (1-S)A^{-\alpha}
  \end{equation}
  with
  $$
  \| S \|\leq C^{(1)}_{\alpha,\beta} \frac{1}{(\inf\spec A)^{\alpha-\beta}} \frac{\left\| |B|^\alpha A^{-\beta} \right\| \left\| |B|^\alpha A^{-\alpha}\right\|^{(1-\alpha)/\alpha}}{1-\left\| |B|^\alpha A^{-\alpha}\right\|^{1/\alpha}} \,.
  $$
  
  In order to prove this, we define, for $t>0$,
  $$
  Y(t) := (A+t)^{-\alpha} |B| (A+t)^{-1+\alpha} \,.
  $$  
  
Let us show that the operators $1\pm Y(t)$ are invertible for all $t>0$. We have
  \begin{align*}
    \|Y(t)\| & \leq \|(A+t)^{-\alpha} |B|^\alpha\| \||B|^{1-\alpha} (A+t)^{-1+\alpha} \| \\
             & \leq \|A^{-\alpha} |B|^\alpha\| \||B|^{1-\alpha} A^{-1+\alpha} \| \\
             & \leq \|  |B|^\alpha A^{-\alpha} \|^{1/\alpha} \,.
  \end{align*}
  In the last step we used
  $|B|^{2\alpha} \leq \||B|^\alpha A^{-\alpha}\|^2 A^{2\alpha}$ and the operator
  monotonicity of $x\mapsto x^{(1-\alpha)/\alpha}$ (since $1/2\leq\alpha\leq 1$).
  By assumption \eqref{eq:small}, we have $\|Y(t)\|<1$ and therefore,
  $1\pm Y(t)$ is invertible with
  \begin{equation}
  \label{eq:nzbound}
  \left\| (1\pm Y(t))^{-1} \right\| \leq (1-\|Y(t)\|)^{-1} \leq \left(1 - \| |B|^\alpha A^{-\alpha} \|^{1/\alpha} \right)^{-1} \,.
  \end{equation}
  
The definition of $Y(t)$ and the invertibility of $1\pm Y(t)$ implies that for $t>0$,
  $$
  (A+B+t)^{-1} = (A+t)^{-1} \mp (A+t)^{-1+\alpha}Y(t)(1\pm Y(t))^{-1} (A+t)^{-\alpha} \,,
  $$
  where the upper sign is chosen for $B\geq 0$ and the lower sign for
  $B\leq 0$. We now use the fact that for any number $h>0$
  $$
  h^{-\alpha} = c_\alpha \int_0^\infty (h+t)^{-1} t^{-\alpha} \,\dt
  \qquad\text{with}\ c_\alpha= \pi^{-1}\sin(\pi\alpha) \,.
  $$
  (Here we use $\alpha<1$.) Therefore, by the spectral theorem,
  \eqref{eq:abs} holds with
  $$
  S := \pm c_\alpha \int_0^\infty (A+t)^{-1+\alpha} Y(t) (1\pm Y(t))^{-1} (1+tA^{-1})^{-\alpha} t^{-\alpha} \,\dt \,.
  $$
  Clearly,
  \begin{equation}
    \label{eq:sbound}
    \|S\| \leq c_\alpha \int_0^\infty \left\|(A+t)^{-1+\beta}\right\| \left\|(A+t)^{\alpha-\beta}Y(t)\right\| \left\| (1\pm Y(t))^{-1} \right\|  t^{-\alpha} \,\dt \,.
  \end{equation}
  We bound the three terms on the right side separately. For the last factor we use \eqref{eq:nzbound}. Next, we bound
  \begin{align*}
    \left\| (A+t)^{\alpha-\beta}Y(t) \right\| & \leq \left\| (A+t)^{-\beta} |B|^\alpha\right\| \left\|  (A+t)^{-1+\alpha} |B|^{1-\alpha} \right\| \\
                                              & \leq \left\| A^{-\beta} |B|^\alpha \right\| \left\| A^{-1+\alpha} |B|^{1-\alpha} \right\| \\
                                              & \leq \left\| |B|^\alpha A^{-\beta} \right\| \left\| |B|^{\alpha} A^{-\alpha} \right\|^{(1-\alpha)/\alpha}.
  \end{align*}
  Finally,
  $$
  \left\| (A+t)^{-1+\beta} \right\| \leq (\inf\spec A +t)^{-1+\beta} \,.
  $$
  Inserting these bounds into \eqref{eq:sbound} we find
  \begin{align*}
    \|S\| & \leq c_\alpha \frac{\left\| |B|^\alpha A^{-\beta} \right\| \left\| |B|^{\alpha} A^{-\alpha} \right\|^{(1-\alpha)/\alpha}}{1-\| |B|^\alpha A^{-\alpha} \|^{1/\alpha}} \int_0^\infty \frac{\dt}{t^\alpha (\inf\spec A +t)^{1-\beta}} \\
          & = C_{\alpha,\beta}^{(1)} \frac{1}{(\inf\spec A)^{\alpha-\beta}} \frac{\left\|  |B|^\alpha A^{-\beta} \right\| \left\| |B|^{\alpha} A^{-\alpha} \right\|^{(1-\alpha)/\alpha}}{1-\| |B|^\alpha A^{-\alpha} \|^{1/\alpha}} \,.
  \end{align*}
  This proves the claim in Step 1.
  
  \medskip

  \emph{Step 2.} We now prove the statement of the lemma by applying
  Step 1 with $A$ replaced by $A+M$ with a sufficiently large constant $M$.

  We first note that
  \begin{equation}
    \label{eq:ass1}
    \left\| |B|^\alpha (A+M)^{-\beta} \right\| \leq \left\| |B|^\alpha A^{-\beta} \right\| <\infty \,.
  \end{equation}
  Moreover, we claim that
  \begin{equation}
    \label{eq:ass2}
    \left\| |B|^\alpha (A+M)^{-\alpha} \right\| \leq C_{\alpha,\beta}^{(2)} \frac{\left\| |B|^\alpha A^{-\beta} \right\|}{M^{\alpha-\beta}} \,.
  \end{equation}
  In fact, by the spectral theorem,
  \begin{align}
    \label{eq:alphabeta}
    |B|^{2\alpha} & \leq \left\| |B|^\alpha A^{-\beta} \right\|^2 A^{2\beta} \notag \\
                  & \leq \left\| |B|^\alpha A^{-\beta} \right\|^2 \left( \sup_{a\geq 0} \frac{a^{2\beta}}{(a+M)^{2\alpha}} \right) (A+M)^{2\alpha} \notag \\
                  & = \left( \left\| |B|^\alpha A^{-\beta} \right\| \frac{C_{\alpha,\beta}^{(2)}}{M^{\alpha-\beta}} \right)^2 (A+M)^{2\alpha} \,,
  \end{align}
  which proves \eqref{eq:ass2}.
  
  It follows from \eqref{eq:ass2} that assumption \eqref{eq:small} in
  Step 1 (with $A$ replaced by $A+M$) is satisfied if
  $$
  M > \left( C_{\alpha,\beta}^{(2)} \left\| |B|^\alpha A^{-\beta} \right\| \right)^{1/(\alpha-\beta)} \,,
  $$
  which we assume henceforth. Inequality \eqref{eq:alphabeta} together
  with the operator monotonicity of $x\mapsto x^{1/2\alpha}$ (since
  $\alpha\geq 1/2$) implies that, if $B\leq 0$,
  $$
  B \geq - \left( \left\| |B|^\alpha A^{-\beta} \right\| \frac{C_{\alpha,\beta}^{(2)}}{M^{\alpha-\beta}} \right)^{1/\alpha} (A+M)
  $$
  and therefore
  $$
  A+M+B \geq \left( 1 - \left( \left\| |B|^\alpha A^{-\beta} \right\| \frac{C_{\alpha,\beta}^{(2)}}{M^{\alpha-\beta}} \right)^{1/\alpha} \right) (A+M) >0 \,.
  $$
  This shows that assumption \eqref{eq:invertible} in Step 1 (with $A$
  replaced by $A+M$) is satisfied.
  
  Applying the result there, we find that there is an operator $S_M$
  (which is defined as before but with $A$ replaced by $A+M$) with
  $$
  (A+M+B)^{-\alpha} = (1-S_M) (A+M)^{-\alpha} = (A+M)^{-\alpha} (1-S_M^*)
  $$
  such that
  $$
  \left\| S_M \right\| \leq \frac{C_{\alpha,\beta}^{(1)}}{M^{(\alpha-\beta)/\alpha}} \frac{\left( C_{\alpha,\beta}^{(2)}\right)^{(1-\alpha)/\alpha} \left\| |B|^\alpha A^{-\beta} \right\|^{1/\alpha}}{1 - \left(C_{\alpha,\beta}^{(2)}\right)^{1/\alpha} \left\| |B|^\alpha A^{-\beta} \right\|^{1/\alpha} M^{-(\alpha-\beta)/\alpha}} \,.
  $$
  Thus, there is a constant $C$, depending only on $\alpha$ and $\beta$,
  such that
  $$
  \|S_M\| \leq 1 - \frac{1}{\sqrt 2}
  \qquad\text{if}\ M \geq C\left\| |B|^\alpha A^{-\beta} \right\|^{1/(\alpha-\beta)}\,,
  $$
  and therefore, since $1-1/\sqrt 2<\sqrt 2 - 1$,
  \begin{align*}
    (A+B+M)^{-2\alpha} & = (A+M)^{-\alpha} (1-S_M^*)(1-S_M) (A+M)^{-\alpha} \\ & \leq \|1-S_M\|^2 (A+M)^{-2\alpha} \\
                       & \leq (1 + \|S_M\|)^2 (A+M)^{-2\alpha} \\
                       & \leq 2 (A+M)^{-2\alpha} \,.
  \end{align*}
  Moreover,
  \begin{align*}
    (A+M)^{-2\alpha} & = (A+B+M)^{-\alpha} (1-S_M^*)^{-1} (1-S_M)^{-1} (A+B+M)^{-\alpha} \\
                     & \leq \left\| \left( 1- S_M \right)^{-1} \right\|^2 (A+B+M)^{-2\alpha} \\
                     & \leq (1- \|S_M\|)^{-2} (A+B+M)^{-2\alpha} \\
                     & \leq 2 (A+B+M)^{-2\alpha} \,.
  \end{align*}
  This concludes the proof of the lemma.
\end{proof}

%%%%%%%%%%%%%%%%%%%

\section{Differentiability for fixed angular momentum}
\label{s:singlel}

Our goal in this section is to prove Proposition \ref{diff1} about the differentiability of $\lambda\mapsto\tr (C_\ell^H-\lambda U)_-$ at $\lambda=0$. Our strategy is to deduce this from the abstract results in the previous section with $A=C_\ell^H$ and $B=U$. Therefore this section is mostly concerned with verifying the assumptions of Theorem \ref{diffgen0}.

We introduce the notations
\begin{align*}
  p_\ell  := \sqrt{ - \frac{\td^2}{\dr^2} + \frac{\ell(\ell+1)}{r^2}}
  \qquad\text{and}\qquad
  C_{\ell} :=\sqrt{p_\ell^2+1}-1
  \qquad\text{in}\ L^2(\R_+,\dr) \,.
\end{align*}
Throughout we fix a constant $\gamma \in (0,2/\pi)$ and recall that we have defined
$$
C_\ell^H = C_\ell - \gamma r^{-1} \,.
$$

%%%%%%%%%%%%%%%%%%%

\subsection{Removing the Coulomb potential}

The following proposition allows us to remove the Coulomb potential in the operator $(C_\ell^H+M)^{s}$ provided $s$ is not too large. It will be important later that for any $\gamma<2/\pi$ we can choose $s>1/2$.

We recall that the constant $\sigma_\gamma$ was defined after \eqref{eq:defsigma}. The value $\gamma=1/2$ will play a special role in some of the results below and we note that $\sigma_{1/2}=1/2$.

\begin{Prop}\label{hardydomcor}
Let $s\leq 1$ if $\gamma\in(0,1/2)$ and let $s<3/2-\sigma_\gamma$ if $\gamma\in[1/2,2/\pi]$. Then for any $\ell\in\N_0$ and any $M>-\inf\spec C_\ell^H$,
  $$
  \left( C_\ell^H +M \right)^{-s} \left( C_\ell+M \right)^s \qquad\text{and}\qquad
  \left( C_\ell^H +M \right)^{s} \left( C_\ell+M \right)^{-s} \qquad\text{are bounded}.
  $$
\end{Prop}

The proof shows that the operators are bounded uniformly in $\ell$. We will, however, not use this fact.

\begin{proof}
  Since $0\geq \sqrt{p_\ell^2+1}-1-p_\ell\geq -1$, the Kato--Rellich theorem
  implies that
  $\left( C_\ell^H +M \right)^{-1} \left( p_\ell - \gamma r^{-1} +M\right)$
  and $(C_\ell^H+M)\left(p_\ell - \gamma r^{-1} +M\right)^{-1}$ are bounded.
Note that the assumptions of the Proposition imply that $s\leq 1$. Therefore the operator monotonicity of $x\mapsto x^s$ implies that
  $\left( C_\ell^H +M \right)^{-s} \left( p_\ell - \gamma r^{-1} +M\right)^s$
  and
  $\left( C_\ell^H +M \right)^{s} \left( p_\ell - \gamma r^{-1} +M\right)^{-s}$
  are bounded. Thus, it suffices to prove that
  $\left( p_\ell - \gamma r^{-1} +M\right)^{-s}$ $\left( p_\ell+M \right)^s$
  and
  $\left( p_\ell - \gamma r^{-1} +M\right)^{s} \left( p_\ell+M \right)^{-s}$
  are bounded. By \cite[Theorem 1.1]{Franketal2019} we have
  \begin{align*}
    & \left( p_\ell+M \right)^{2s} \leq 2^{(2s-1)_+} \left( p_\ell^{2s} + M^{2s} \right)
      \leq 2^{(2s-1)_+} \left( A_{s,\gamma} \left( p_\ell - \frac{\gamma}{r} \right)^{2s} + M^{2s} \right).
  \end{align*}
  Clearly, the operator on the right side is bounded by a constant times
  $(p_\ell -\gamma r^{-1} +M)^{2s}$.
  This shows that
  $\left( p_\ell - \gamma r^{-1} +M\right)^{-s} \left( p_\ell+M \right)^s$
  is bounded. The proof for
  $\left( p_\ell - \gamma r^{-1} +M\right)^{s} \left( p_\ell+M \right)^{-s}$
  is similar, using also \cite[Theorem 1.1]{Franketal2019}.
\end{proof}

\begin{Rem}\label{elementary}
The above proof relies on \cite{Franketal2019}. However, this machinery is only needed for $1/2\leq\gamma<2/\pi$ and $\ell=0$.  To see this, we recall Hardy's inequality in angular momentum channel $\ell$,
\begin{equation}
\label{eq:hardy}
  \left(f,\frac{(\ell+1/2)^2}{r^2}f\right)=\int_0^\infty\frac{\ell(\ell+1)+1/4}{r^2}|f(r)|^2\,\dr\leq(f,p_\ell^2f) \,.
\end{equation}
This implies that
$$
\left\| \left( p_\ell - \gamma r^{-1} \right) f \right\|
  \geq \left\| p_\ell f\right\| - \gamma \left\| r^{-1} f \right\|
  \geq (1-(\ell+1/2)^{-1} \gamma) \left\| p_\ell f\right\|.
$$
This together with operator monotonicity of $x\mapsto x^s$ for $s\leq 1$ implies that, if $\gamma<\ell+1/2$, then
$$
p_\ell^{2s} \leq (1-(\ell+1/2)^{-1}\gamma)^{-2s} (p_\ell-\gamma r^{-1})^{2s} \,.
$$
Note that the assumption $\gamma<\ell+1/2$ is satisfied for all $\gamma\leq 2/\pi$ if $\ell\geq 1$. Similarly, one shows that
$$
(p_\ell-\gamma r^{-1})^{2s} \leq (1+(\ell+1/2)^{-1}\gamma)^{-2s} p_\ell^{2s} \,.
$$
The previous two bounds yield Proposition \ref{hardydomcor} in the claimed restricted range.
\end{Rem}

\begin{Prop}
  \label{hardydomcoru}
Let $s\leq 1$ if $\gamma\in(0,1/2)$ and let $s<3/2-\sigma_\gamma$ if $\gamma\in[1/2,2/\pi)$. Let $\ell\in\N_0$, $M>-\inf\spec C_\ell^H$ and $0\leq U\in r^{-1} L^\infty((0,\infty))$. Then
  $$
  \limsup_{\lambda\to 0} \left\| \left( C_\ell^H -\lambda U+M \right)^{-s} \left(C_\ell^H +M \right)^s \right\| <\infty \,.
  $$
\end{Prop}

\begin{proof}
Using Proposition \ref{hardydomcor} and the arguments in its proof we see that it is enough to prove that
$$
  \limsup_{\lambda\to 0} \left\| \left( p_\ell -\lambda U+M \right)^{-s} \left(p_\ell +M \right)^s \right\| <\infty \,.
  $$
This follows as in the proof of Proposition \ref{hardydomcor} from \cite[Theorem 4.1]{Franketal2019}.
\end{proof}

%%%%%%%%%%%%%%%%%%%

\subsection{Trace ideal bounds}

As we already mentioned, we denote by $L^p_{\rm c}([0,\infty))$ the space of all functions in $L^p$ whose support is a compact subset of $[0,\infty)$. Moreover, we denote by $L^p_{\rm loc}((0,\infty))$ the space of all functions which are in $L^p$ on any compact subset of $(0,\infty)$. We now introduce the test function space which appears in Theorem \ref{convfixedl} and Proposition \ref{diff1},
\begin{align}
  \label{eq:deftestnod}
  \begin{split}
    \cd_\gamma^{(0)} & :=
\begin{cases}
  & \{ W\in L^1_{\rm loc}((0,\infty)):\ W \in L^{2s}((0,\infty),\min\{r^{2s'-1},1\}\dr)\cap L^1((1,\infty))  \\
  & \qquad\qquad\qquad\qquad\quad\ \text{for some}\ 1/2<s'<s\leq 1 \} \\
    & \qquad\qquad\qquad \text{if}\ 0<\gamma<1/2 \,,
\\
  & \{ W\in L^1_{\rm loc}((0,\infty)):\ W \in L^{2s}((0,\infty),\min\{r^{2s'-1},1\}\dr)\cap L^1((1,\infty))  \\
  & \qquad\qquad\qquad\qquad\quad\ \text{for some}\ 1/2<s'<s <3/2-\sigma_\gamma \} \\
    & \qquad\qquad\qquad \text{if}\ 1/2\leq\gamma<2/\pi \,.
  \end{cases}
  \end{split}
\end{align}

It will be convenient to introduce another class of function spaces, namely, for $s\geq 1/2$ we define
\begin{align}
  \label{eq:defsigmadeltanod}
  \begin{split}
    \ck_{s}^{(0)} & :=\left\{W \in L^1_{\rm loc}((0,\infty)):\ \|W\|_{\ck_{s}^{(0)}}<\infty\right\} \,, \\
    \|W\|_{\ck_{s}^{(0)}} & := \int_0^1 r^{2s-1}|W(r)|\,\dr + \int_1^\infty |W(r)|\,\dr \,.
      \end{split}
\end{align}
The connection between these spaces and $\cd_\gamma^{(0)}$ is that
\begin{align}
  \label{eq:deftestnod0}
  \begin{split}
    \cd_\gamma^{(0)} =
\begin{cases}
  & \{ W\in\ck_{s}^{(0)}:\  |W|^{2s}\in\ck_{s'}^{(0)} \ \text{for some}\ 1/2<s'<s\leq 1 \} \\
    & \qquad\qquad\qquad \text{if}\ 0<\gamma<1/2 \,.
\\
  & \{ W\in\ck_{s}^{(0)}:\  |W|^{2s}\in\ck_{s'}^{(0)}  \ \text{for some}\ 1/2<s'<s<3/2-\sigma_\gamma \} \\
  & \qquad\qquad\qquad \text{if}\ 1/2\leq\gamma<2/\pi \,.
  \end{cases}
  \end{split}
\end{align}
Indeed, the inclusion $\supset$ is clear and for $\subset$ it suffices to note that
\begin{align}
\label{eq:holder}
\begin{split}
\int_0^1 r^{2s-1} |W|\,\dr & \leq \left( \int_0^1 r^{2s'-1}|W|^{2s}\,\dr \right)^{1/(2s)} \left(\int_0^1 r^{2s-1 + \tfrac{2(s-s')}{2s-1}} \,\dr \right)^{(2s-1)/(2s)} \\
& = A_{s,s'} \left( \int_0^1 r^{2s'-1}|W|^{2s}\,\dr \right)^{1/(2s)}.
\end{split}
\end{align}

The spaces $\ck_s^{(0)}$ appear naturally in the following trace ideal bound. We denote the Hilbert--Schmidt norm by $\|\cdot\|_2$.

\begin{Prop}
  \label{genReltrclassnod}
  Let $s\in(1/2,1]$, $\ell\in\N_0$ and $M>0$. Then for all $0\leq W \in\ck_{s}^{(0)}$,
  \begin{align}
    \label{eq:genReltrclassnod}
    \|W^{1/2}(C_\ell+M)^{-s}\|_2^2 \leq A_{s,\ell,M} \|W\|_{\ck_{s}^{(0)}}\,.
  \end{align}
  In particular,
  \begin{align}
    \label{eq:genSobolevDualnod}
    W \leq A_{s,\ell,M} \|W\|_{\ck_{s}^{(0)}} (C_\ell+M)^{2s}\,.
  \end{align}
\end{Prop}

\begin{proof}
We recall the definition of the Fourier--Bessel (or Hankel) transform $\Phi_\ell$,
$$
(\Phi_\ell f)(k):= i^\ell \int_0^\infty (kr)^{1/2}J_{\ell+1/2}(kr)f(r)\,\dr \quad \text{for all}\ \ell\in\N_0\,;
$$
see e.g., \cite[(B.105)]{Messiah1969}. It is well known that for each
$\ell\in\N_0$, $\Phi_\ell$ is unitary on $L^2(\R_+)$ and, moreover, it
diagonalizes $p_\ell^2$ in the sense that for any $f$ from the domain of
this operator,
$$
(\Phi_\ell p_\ell^2 f)(k) = k^2 (\Phi_\ell f)(k)\,.
$$
This implies
  \begin{align*}
    \|W^{1/2}(C_\ell+M)^{-s}\|_2^2 & = \int_0^\infty \dr\, W(r)\int_0^\infty \dk\, \frac{kr J_{\ell+1/2}(kr)^2}{(\sqrt{k^2+1}-1+M)^{2s}}\\
                                   & \leq A_{s,\ell,M} \|W\|_{\ck_{s}^{(0)}}\,.
  \end{align*}
The inequality here follows from Lemma \ref{auxboundnod} in the appendix. Inequality \eqref{eq:genSobolevDualnod} follows from \eqref{eq:genReltrclassnod} since the Hilbert--Schmidt norm does not exceed the operator norm.
\end{proof}

The following bound for the Chandrasekhar hydrogen operator is an immediate consequence of Propositions \ref{hardydomcor} and \ref{genReltrclassnod}.

\begin{Cor}
  \label{genReltrclassHydro}
Let $s\in(1/2,1]$ if $\gamma\in(0,1/2)$ and let $s\in(1/2,3/2-\sigma_\gamma)$ if $\gamma\in[1/2,2/\pi)$. Let $\ell\in\N_0$ and $M>-\inf\spec C_\ell^H$. Then for all $0\leq W \in\ck_{s}^{(0)}$,
  \begin{align}
    \label{eq:genReltrclassHydro}
    \|W^{1/2}(C_\ell^H+M)^{-s}\|_2^2 \leq A_{\gamma,s,\ell,M} \|W\|_{\ck_{s}^{(0)}}\,.
  \end{align}
  In particular,
  \begin{align}
    \label{eq:genSobolevDualHydro}
    W \leq A_{\gamma,s,\ell,M} \|W\|_{\ck_{s}^{(0)}} (C_\ell^H+M)^{2s}\,.
  \end{align}
\end{Cor}

%%%%%%%%%%%%%%

%%%%%%%%%%%%%%%

\subsection{Proof of Proposition \ref{diff1}}

We now prove the main result of this section. Let $\gamma\in(0,2/\pi)$, $\ell_0\in\N_0$ and let $U$ be a non-negative function on $(0,\infty)$ which belongs either to $r^{-1} L^\infty_{\rm c}([0,\infty))$ or to $\mathcal D_\gamma^{(0)}$. We will apply Theorem \ref{diffgen0} with $A=C_{\ell_0}^H$ and $B=U$.

The fact that $A_-$ is trace class was shown in \cite[Lemma 1]{Franketal2008}. Moreover, the fact that zero is not an eigenvalue of $C^H$, and consequently not of $C^H_{\ell_0}$, was shown by Herbst in \cite[Theorem 2.3]{Herbst1977}.

Now let us assume first that $U\in r^{-1} L^\infty_{\rm c}$. Then $B$ is form bounded with respect to $A$ by Kato's inequality and, by a simple computation, $U\in\ck_{s}^{(0)}$ for any $s>1/2$. Thus, assumption \eqref{eq:traceclassdelta0} follows from Corollary \ref{genReltrclassHydro}. Moreover, assumption \eqref{eq:relbounddelta} follows from Proposition \ref{hardydomcoru}.

On the other hand, if $U\in \mathcal D_\gamma^{(0)}$, then by \eqref{eq:deftestnod0} there are $1/2<s'<s\leq 1$ if $\gamma<1/2$ and $1/2<s'<s<3/2-\sigma_\gamma$ if $1/2\leq\gamma<2/\pi$ such that $U\in\ck_{s}^{(0)}$ and $U^{2s}\in\ck_{s'}^{(0)}$. Thus, Assumption \eqref{eq:traceclassdelta0} follows again from Corollary \ref{genReltrclassHydro}. Moreover, Corollary~\ref{genReltrclassHydro} with $s'$ instead of $s$ and with $U^{2s}$ instead of $W$ implies that \eqref{eq:increasedrelbdd} holds, and therefore assumption \eqref{eq:relbounddelta} follows from Proposition \ref{diffgen}.
\qed

%%%%%%%%%%%%%%%

\section{Controlling large angular momenta}
\label{s:alll}

\subsection{Test functions}
We define
\begin{align}
\label{eq:deftest}
\begin{split}
\cd := & \left\{ W \in L^1_{\rm loc}((0,\infty)):\ \sup_{R\geq 1} R^{(4s-1)/2} \left( \int_R^{2R} |W|\,\dr + \int_R^{2R} |W|^{2s} \,\dr \right) <\infty \,, \right. \\
& \qquad\qquad \qquad\qquad\quad \left. \int_0^1 r^{2s'-1} |W|^{2s}\,\dr <\infty \ \text{for some}\ 1/2< s'<s\leq 3/4 \right\}.
\end{split}
\end{align}
It will be convenient to introduce another class of function spaces, namely, for $s\geq1/2$ and $\delta\in[0,2s-1]$ we define
\begin{align}
  \label{eq:defsigmadelta}
  \begin{split}
    \ck_{s,\delta} & :=\{W \in L^1_{\rm loc}((0,\infty)):\ \|W\|_{\ck_{s,\delta}}<\infty\} \,, \\
    \|W\|_{\ck_{s,\delta}} & :=\max\left\{ \int_0^1 r^{2s-1}|W(r)|\,\dr \,,\, \sup_{R\geq 1} R^{(\delta+4s-1)/2} \int_R^{2R} |W(r)|\,\dr\right\}.
  \end{split}
\end{align}
The connection between these spaces and $\cd$ is that
\begin{align}
  \label{eq:deftest0}
  \begin{split}
    \cd & = \{W \in\ck_{s,0}\,:\  |W|^{2s}\in\ck_{s',4(s-s')}\\
    & \qquad\qquad\qquad\qquad \text{for some}\ 1/2<2s/3+1/6\leq s'<s\leq 3/4 \}\,.
  \end{split}
\end{align}
Note that the assumption $s'\geq 2s/3+1/6$ ensures that $4(s-s')\leq 2s'-1$, as required in the definition of $\ck_{s',4(s-s')}$. The proof of \eqref{eq:deftest0} uses again \eqref{eq:holder}.

While the above definition is convenient to verify whether a given function belongs to $\ck_{s,\delta}$, in our proofs we will use an equivalent characterization given in the following lemma.

\begin{Lem}\label{equivnorm}
Let $s\geq 1/2$ and $\delta\in[0,2s-1]$. Then
\begin{align*}
\|W\|_{\ck_{s,\delta}} \sim \sup_{R\geq 1} R^\delta & \left[\int_0^R \left(\frac{r}{R}\right)^{2s-1}|W(r)|\,\dr + \int_R^{R^2}\left(\frac{r}{R}\right)^{4s-1} |W(r)|\,\dr\right.\\
 & \quad \left.+ R^{4s-1}\int_{R^2}^\infty |W(r)|\,\dr\right].
\end{align*}
\end{Lem}

\begin{proof}
Let us denote the supremum appearing in the lemma by $[W]$. We have, if $R\geq 2$,
$$
[W] \geq R^\delta \int_{R^2/2}^{R^2} (r/R)^{4s-1}|W(r)|\,\dr \geq 2^{-4s-1} R^{\delta+4s-1} \int_{R^2/2}^{R^2} |W(r)|\,\dr \,.
$$
Combining this with the bounds $[W]\geq \int_0^1 r^{2s-1}|W(r)|\,\dr$ and, if $1\leq R\leq 2$,
\begin{align*}
[W] & \geq 4^{\delta-2s+1} \int_0^4 r^{2s-1}|W(r)|\,\dr \geq 4^{\delta-2s+1} \int_R^{2R} r^{2s-1}|W(r)|\,\dr \\
& \geq 4^{(3\delta-8s+3)/4} R^{(\delta+4s-1)/2} \int_R^{2R} |W(r)|\,\dr \,,
\end{align*}
we conclude that $[W]\gtrsim \|W\|_{\ck_{s,\delta}}$.

Conversely, if $R\geq 1$, then
\begin{align*}
\int_0^R \left(\frac{r}{R}\right)^{2s-1}|W|\,\dr & \leq \int_0^1 \left(\frac{r}{R}\right)^{2s-1} |W|\,\dr + \sum_{0\leq k<\log_2 R} \int_{2^{k}}^{2^{k+1}} \left(\frac{r}{R}\right)^{2s-1}|W|\,\dr \\
& \leq R^{-2s+1} \|W\|_{\ck_{s,\delta}} \left( 1 + \! \sum_{0\leq k<\log_2 R} \! 2^{(k+1)(2s-1)} 2^{-k(\delta+4s-1)/2} \right) \\
& \lesssim R^{-2s+1} \|W\|_{\ck_{s,\delta}} \leq R^{-\delta} \|W\|_{\ck_{s,\delta}}
\end{align*}
and, since $4s-1>2s-1\geq\delta$,
\begin{align*}
\int_R^{R^2} \left(\frac{r}{R}\right)^{4s-1}|W|\,\dr & \leq \sum_{0\leq k<\log_2 R} \int_{2^{k} R}^{2^{k+1} R} \left(\frac{r}{R}\right)^{4s-1}|W|\,\dr \\
& \leq \|W\|_{\ck_{s,\delta}} \sum_{0\leq k<\log_2 R} 2^{(k+1)(4s-1)} (2^{k} R)^{-(\delta+4s-1)/2} \\
& \lesssim R^{-\delta} \|W\|_{\ck_{s,\delta}} \,.
\end{align*}
Finally, using $\delta+4s-1>0$,
\begin{align*}
\int_{R^2}^\infty |W|\,\dr & = \sum_{k\geq 0} \int_{2^{k}R^2}^{2^{k+1} R^2} |W|\,\dr \leq \|W\|_{\ck_{s,\delta}} \sum_{k\geq 0} (2^{k}R^2)^{-(\delta+4s-1)/2} \\
& \lesssim R^{-(\delta+4s-1)} \|W\|_{\ck_{s,\delta}}.
\end{align*}
Thus, $[W]\lesssim \|W\|_{\ck_{s,\delta}}$, as claimed.
\end{proof}

%%%%%%%%%%%%%

\subsection{Proof of Proposition \ref{diff2}}

The following proposition is the analogue of Proposition \ref{genReltrclassnod} where $M$ is chosen in an specific $\ell$-dependent manner and where we track the dependence of the constant on $\ell$.

\begin{Prop}
  \label{genReltrclass}
  Let $s\in(1/2,3/4]$, $\delta\geq 0$ and $a>0$. Then for all $\ell\in\N_0$ and all $0\leq W \in\ck_{s,\delta}$,
  \begin{align}
    \label{eq:genReltrclass}
    \|W^{1/2}(C_\ell+a(\ell+1/2)^{-2})^{-s}\|_2^2 \leq A_{s,a} (\ell+1/2)^{-\delta}\|W\|_{\ck_{s,\delta}}\,.
  \end{align}
  In particular,
  \begin{align}
    \label{eq:genSobolevDual}
    W \leq A_{s,a} (\ell+1/2)^{-\delta}\|W\|_{\ck_{s,\delta}} (C_\ell+a(\ell+1/2)^{-2})^{2s}\,.
  \end{align}
\end{Prop}

\begin{proof}
  Using the Fourier--Bessel transform as in the proof of Proposition \ref{genReltrclassnod}, we have
  \begin{align*}
    & \|W^{1/2}(C_\ell+a(\ell+1/2)^{-2})^{-s}\|_2^2\\
    & \quad = \int_0^\infty \dr\, W(r)\int_0^\infty \dk\, \frac{kr J_{\ell+1/2}(kr)^2}{(\sqrt{k^2+1}-1+a(\ell+1/2)^{-2})^{2s}}\\
    & \quad \leq A_s(\ell+1/2)^{-\delta}\|W\|_{\ck_{s,\delta}}\,.
  \end{align*}
The inequality here follows from Lemma \ref{auxbound} in the appendix and the characterization of the norm in $\ck_{s,\delta}$ in Lemma \ref{equivnorm}. Inequality \eqref{eq:genSobolevDual} follows from \eqref{eq:genReltrclass} since the Hilbert--Schmidt norm bounds the operator norm.
\end{proof}

The next proposition implies, in particular, a lower bound on the lowest eigenvalue of $C_\ell^H-\lambda U$ for sufficiently small $\lambda$, which generalizes the result from \cite[Theorem 2.2]{Franketal2009} for $\lambda=0$.

\begin{Prop}
  \label{boundA}
  Let $0<\gamma<2/\pi$ and $s\in(1/2,3/4]$. Then there are constants $a_\gamma,c_{\gamma,s}<\infty$ such that for all $\ell\in\N_0$, all measurable functions $V, U$ on $(0,\infty)$ with $0\leq V(r)\leq\gamma/r$ and $|U|^{2s}\in \ck_{s,0}$ and all $\lambda\in\R$ with $|\lambda|\leq c_{\gamma,s} \| |U|^{2s}\|_{\ck_{s,0}}^{-1/(2s)}$ one has
  $$
  C_\ell-V-\lambda U \geq -a_{\gamma}(\ell+1/2)^{-2} \,.
  $$
\end{Prop}

\begin{proof}
Clearly, we may assume $V=\gamma/r$. We fix a number $\theta\in(0,1)$ such that $\gamma\leq (1-\theta)2/\pi$ and write
$$
C_\ell-V-\lambda U = (1-\theta) \left( C_\ell - (1-\theta)^{-1} V \right) + \theta \left( C_\ell -(\lambda/\theta) U \right).
$$
From \cite[Theorem 2.2]{Franketal2009} we know that there is a constant $a_\gamma$ such that $C_\ell - (1-\theta)^{-1} V \geq -a_\gamma (\ell+1/2)^{-2}$ for all $\ell\in\N_0$. On the other hand, from \eqref{eq:genSobolevDual} with $a=a_\gamma$ and operator monotonicity of $x\mapsto x^{s}$ we know that for all $\ell\in\N_0$,
  $$
  |U| \leq A_{s,a_\gamma}^{1/(2s)} \| |U|^{2s}\|_{\ck_{s,0}}^{1/(2s)} (C_\ell+a_\gamma(\ell+1/2)^{-2} ) \,.
  $$
Thus, if
$$
|\lambda| \theta^{-1} A_{s,a_\gamma}^{1/(2s)} \| |U|^{2s}\|_{\ck_{s,0}}^{1/(2s)} \leq 1 \,,
$$  
then $C_\ell -(\lambda/\theta) U\geq -a_\gamma(\ell+1/2)^{-2}$ for all $\ell\in\N_0$, as claimed.
\end{proof}

Finally, we are in position to give the

\begin{proof}[Proof of Proposition \ref{diff2}]
By assumption and \eqref{eq:deftest0}, we have $U=U_1+U_2$ with $0\leq U_1\in r^{-1}L^\infty_{\rm c}([0,\infty))$ and $0\leq U_2 \in \ck_{s,0}$ such that $U_2^{2s}\in \ck_{s',4(s-s')}$ for some $1/2<2s/3+1/6\leq s'<s\leq 3/4$. In case $U_2=0$ we choose an arbitrary number $1/2<s\leq 3/4$.

We set $V:=\gamma/r -\chi$ and denote by $d_{\ell,\lambda}$ the orthogonal projection onto the negative spectral subspace of $C_\ell-V-\lambda U$. Then, by the variational principle,
  $$
\tr(C_\ell-V-\lambda U)_- - \tr(C_\ell-V)_-\leq \lambda \tr(d_{\ell,\lambda}U) 
= \lambda\tr ABCB^*A^*
  $$
  with
  \begin{align*}
    A&:=d_{\ell,\lambda}(C_\ell-V-\lambda U+b_\ell)^s\\
    B&:=(C_\ell-V-\lambda U+b_\ell)^{-s}(C_\ell+b_\ell)^s\\
    C&:=(C_\ell+b_\ell)^{-s}U(C_\ell+b_\ell)^{-s}\,.
  \end{align*}
We pick $s$ as in the assumption on $U_2$ and
  $$
  b_\ell:=\frac{a}{(\ell+1/2)^2}
  $$
with $a$ to be determined.

We fix $\gamma'\in(\gamma,2/\pi)$ and apply Proposition \ref{boundA} with $V+\lambda U_1$ in place of $V$, $U_2$ in place of $U$ and $\gamma'$ in place of $\gamma$. Note that $U_2^{2s}\in \ck_{s',4(s-s')}\subset \ck_{s,0}$. We infer that there are $\lambda_1>0$ and $a_\gamma$ such that for all $\ell\in\N_0$ and all $|\lambda|\leq\lambda_1$ one has $C_\ell - V -\lambda U + b_\ell \geq (a-a_\gamma)(\ell+1/2)^{-2}$. We assume from now on that $a\geq a_\gamma$, so that the operator $A$ is well-defined, and we obtain
$$
\|A\|^2\leq a^{2s}(\ell+1/2)^{-4s} \,.
$$

  We now turn to $B = B_1 B_2$ with
  \begin{align*}
    B_1 & := (C_\ell-V-\lambda U+b_\ell)^{-s}(C_\ell-V-\lambda U_1+b_\ell)^s \,,\\
    B_2 & := (C_\ell-V-\lambda U_1+b_\ell)^{-s}(C_\ell+b_\ell)^s \,,
  \end{align*}
  and show that there are $a<\infty $ and $\ell_*\in\N_0$ such that for all $\ell\geq\ell_*$ one has $\|B_1\|^2 \leq 2$ and $\|B_2\|^2 \leq 4^s$.
  
  We begin with the bound on $B_2$. We will show that there are $a<\infty $ and $\ell_*\in\N_0$ such that for all $\ell\geq\ell_*$,
\begin{equation}
\label{eq:proofb2}
  (C_\ell+b_\ell)^{2}\leq 4 (C_\ell-V-\lambda U_1+b_\ell)^{2} \,,
\end{equation}
  which by operator monotonicity of $x\mapsto x^s$ implies that
  $$
\| B_2 \|^2 \leq 4^s \,.
  $$
For the proof of \eqref{eq:proofb2} we use an argument similar to that in Remark \ref{elementary}. We observe
\begin{equation}
\label{eq:plancherel}
\|(p_\ell+b)(C_\ell+b)^{-1}\| = \sup_{k\geq 0} \frac{k+b}{\sqrt{k^2+1}-1+b} \leq C \left( b^{-1/2} \one_{\{b\leq 1\}} + \one_{\{b>1\}} \right).
\end{equation} 
This can be seen either by an explicit computation of the supremum or by straightforward bounds in the three regions $k\leq\min\{b,1\}$, $\min\{b,1\}<k\leq\max\{b,1\}$ and $k>\max\{b,1\}$. Together with Hardy's inequality \eqref{eq:hardy} we obtain
\begin{align*}
\|(V+\lambda U_1)f\| & \leq C \, (\gamma+|\lambda|\|rU_1\|_\infty)
\left( a^{-1/2} \one_{\{a\leq (\ell+\frac12)^2\}} + (\ell+\tfrac 12)^{-1} \one_{\{a>(\ell+\frac 12)^2\}} \right) \\
& \quad \times \|(C_\ell+b_\ell)f\|\,.
\end{align*}
If we choose
$$
a = 4 C^2 (\gamma+\lambda_1 \|r U_1\|_\infty)^2
\qquad\text{and}\qquad
\ell_* = \left\lfloor \sqrt a + \frac12 \right\rfloor,
$$
then the previous inequality implies that for all $|\lambda|\leq\lambda_1$ and $\ell\geq\ell_*$,
$$
\|(V+\lambda U_1)f\| \leq \frac12 \|(C_\ell+b_\ell)f\| \,,
$$
and therefore
  \begin{align*}
    \|(C_\ell-V-\lambda U_1+b_\ell)f\| \geq \frac12 \|(C_\ell+b_\ell)f\|\,.
  \end{align*}
This proves \eqref{eq:proofb2}.

Now we shall show that $\|B_1\|^2\leq 2$ for all $\ell\geq\ell_*$. We will deduce this from Lemma \ref{apriori} with $A=C_\ell-V-\lambda U_1+b_\ell/2$, $B=\lambda  U_2$, $M=b_\ell/2$, $\alpha=s$ and $\beta=s'$, where $s'$ appears in the assumption on $U_2$. In order to apply that lemma, we need to show that $M\geq C \| |B|^\alpha A^{-\beta} \|^{1/(\alpha-\beta)}$ for a certain constant $C$ depending only on $\alpha$ and $\beta$. For us, this condition takes the form
  \begin{align}
    \label{eq:Neidhardt}
    \begin{split}
      |\lambda|^s \| U_2^s (C_\ell-V-\lambda U_1+b_\ell/2)^{-s'}\| & \leq C^{-s+s'} (b_\ell/2)^{s-s'} \\
      & = (2C)^{-s+s'}a^{s-s'} (\ell+1/2)^{-2(s-s')} \,.
    \end{split}
  \end{align}
In order to prove this, we combine \eqref{eq:genSobolevDual} with the $s'$-th root of \eqref{eq:proofb2} to obtain for all $\ell\geq\ell_*$,
  \begin{align*}
    U_2^{2s}
    & \leq A_{s'} (\ell+1/2)^{-4(s-s')} \|U_2^{2s}\|_{\ck_{s',4(s-s')}}(C_\ell+b_\ell)^{2s'}\\
    & \leq 4^{s'} A_{s'} (\ell+1/2)^{-4(s-s')}\|U_2^{2s}\|_{\ck_{s',4(s-s')}}(C_\ell-V-\lambda U_1+b_\ell)^{2s'}\,.
  \end{align*}
  Thus, the left side in \eqref{eq:Neidhardt} is bounded by
  \begin{align*}
    |\lambda|^s 2^{s'} A_{s'}^{1/2} \|U_2^{2s}\|_{\ck_{s',4(s-s')}}^{1/2} (\ell+1/2)^{-2(s-s')}\,.
  \end{align*}
Consequently, there is a $\lambda_2$ such that \eqref{eq:Neidhardt} is satisfied for all $|\lambda|\leq\lambda_2$ .
    
  Finally, we bound the trace of $C$. By Proposition
  \ref{genReltrclass},
  \begin{align*}
    \tr C = \left\| U^{1/2}(C_\ell+b_\ell)^{-s}\right\|_2^2
    \leq A_{s}\|U\|_{\ck_{s,0}} \,.
  \end{align*}

To summarize, we have shown that for $|\lambda|\leq\min\{\lambda_1,\lambda_2\}$ and $\ell\geq\ell_*$,
$$
\tr ABCB^*A^* \leq \|A\|^2 \|B\|^2 \tr C \leq A_s' \|U\|_{\ck_{s,0}} (\ell+1/2)^{-4s} \,.
$$
Since $s>1/2$, this proves the claimed bound.
\end{proof}

%%%%%%%%%%%%%%%

\subsection{Proof of Theorem \ref{existencerhoh}}

We now prove the pointwise bounds on $\rho_\ell^H$ in Theorem~\ref{existencerhoh}. The proof is a variation of the proof of Proposition \ref{diff2}. We denote by $d_\ell$ the orthogonal projection onto the negative spectral subspace of $C_\ell^H$ and write, similarly as before,
  $$
  \rho_\ell^H(r) = \tr d_\ell \delta_r = \tr ABCB^*A^*
  $$
  with
  \begin{align*}
    A&:=d_{\ell}(C_\ell^H+b_\ell)^s \,,\\
    B&:=(C_\ell^H+b_\ell)^{-s}(C_\ell+b_\ell)^s \,,\\
    C&:=(C_\ell+b_\ell)^{-s}\delta_r(C_\ell+b_\ell)^{-s} \,.
  \end{align*}
Here $\delta_r$ is the delta function at $r\in(0,\infty)$ and $b_\ell=a(\ell+1/2)^{-2}$ for some $a>0$ to be chosen later; $s$ is a parameter satisfying   $1/2<s<3/2-\sigma_\gamma$ and $s\leq 3/4$.

We know from \cite[Theorem 2.2]{Franketal2009} (see also Proposition \ref{boundA}) that there is an $a_0>0$ such that $C_\ell^H \geq -a_0(\ell+1/2)^{-2}$. Thus, for $a\geq a_0$ the operator $A$ is well-defined and we have
$$
\|A\|^2 \leq a^{2s} (\ell+1/2)^{-4s} \,.
$$

The fact that $\|B\|$ is uniformly bounded in $\ell$ was shown in the proof of
  Proposition~\ref{diff2} with a specific choice of $a$ provided $\ell\geq\ell_*$. In this bound we used the assumption $s\leq 1$. For $\ell<\ell_*$ the boundedness of $\|B\|$ follows from Proposition \ref{hardydomcor}. Note that there is no issue with uniformity since we apply this proposition only for a fixed finite number of $\ell$'s. In order to apply Proposition \ref{hardydomcor} we need the additional assumption $s<3/2-\sigma_\gamma$ if $\gamma\geq 1/2$.

Finally, by the Fourier--Bessel transform as in the proof of Proposition \ref{genReltrclassnod} and by Lemma \ref{auxbound} we have
  \begin{align*}
  \tr C & = (C_\ell+a_\ell)^{-2s}(r,r) \\
  & \leq A_{s,a}\left[\left(\frac{r}{\ell+\frac12}\right)^{2s-1}\one_{\{r\leq\ell+\frac12\}}+\left(\frac{r}{\ell+\frac12}\right)^{4s-1}\one_{\{\ell+\frac12\leq r\leq(\ell+\frac12)^2\}}\right.\\
    & \qquad\qquad \left.+ \left(\ell+\frac12\right)^{4s-1}\one_{\{r\geq(\ell+1/2)^2\}}\right].
  \end{align*}
Here we used the assumption $1/2<s\leq 3/4$.

To summarize, we have shown that for all $\ell\in\N_0$,
  \begin{align*}
& \tr ABCB^*A^* \leq \|A\|^2 \|B\|^2 \tr C \\
    & \quad \leq A_{s,\gamma}(\ell+1/2)^{-4s}\left[\left(\frac{r}{\ell+\frac12}\right)^{2s-1}\one_{\{r\leq\ell+\frac12\}}+\left(\frac{r}{\ell+\frac12}\right)^{4s-1}\one_{\{\ell+\frac12\leq r\leq(\ell+\frac12)^2\}}\right.\\
    & \qquad\qquad\qquad\qquad\qquad \left.+ \left(\ell+\frac12\right)^{4s-1}\one_{\{r\geq(\ell+1/2)^2\}}\right].
  \end{align*}
This proves the first assertion in the theorem.  To obtain the second assertion we recall \eqref{eq:totaldenshydro}. By summing the bounds from the first part of the theorem we obtain
$$
\rho^H(r) \leq A_{s,\gamma} \left( r^{2s-3} \one_{\{r\leq 1\}} + r^{-3/2} \one_{\{r>1\}} \right).
$$
Recalling the assumptions on $s$ we obtain the claimed bound.
\qed

%%%%%%%%%%%

\subsection{An improvement of Theorem \ref{existencerhoh}}

The following theorem complements and improves the bound on $\rho^H(r)$ for small $r$ when $\gamma<(1+\sqrt 2)/4$.

\begin{Thm}
  \label{densityagain}
  Let $3/4<s\leq 1$ if $0<\gamma<1/2$ and $1/2<s<3/2-\sigma_\gamma$ if $1/2\leq\gamma<(1+\sqrt 2)/4$. Then for all $\ell\in\N_0$ and $r\in\R_+$
  \begin{align*}
    \rho_\ell^H(r)
    & \leq A_{s,\gamma} \!\left( \ell+\tfrac 12 \right)^{-4s} \!
    \left[ \! \left(\frac{r}{\ell+\tfrac12}\right)^{\!\! 2s-1} \!\! \one_{\{r\leq (\ell+\frac12)^\alpha \}} 
   \! + \! \frac{r}{(\ell+\tfrac12)^{4-4s}} \one_{\{(\ell+\frac12)^\alpha < r\leq(\ell+\frac12)^\beta \}}\right.\\    
    & \qquad\qquad\quad\quad \left. + \left(\frac{r}{\ell+\tfrac12}\right)^{\!\! 4s-1}\!\! \one_{\{(\ell+\frac12)^{\beta} < r\leq(\ell+\frac12)^2\}}
    + \left(\ell+\tfrac12\right)^{4s-1}\one_{\{r>(\ell+\frac 12)^2\}}\right]
  \end{align*}
with $\alpha=(5-6s)/(2-2s)$ and $\beta=(8s-5)/(4s-2)$ and with the convention that $(\ell+\frac12)^\alpha=0$ for $s=1$. Moreover, for any $\epsilon>0$ and $r\in\R_+$,
  \begin{align*}
    \rho^H(r) \leq 
\begin{cases}
A_\gamma \left( r^{-1} \one_{\{r\leq 1\}} + r^{-3/2} \one_{\{r>1\}} \right)
 & \text{if}\ 0<\gamma<1/2 \,, \\
A_{\gamma,\epsilon} \left( r^{-2\sigma_\gamma-\epsilon} \one_{\{r\leq 1\}} + r^{-3/2} \one_{\{r>1\}} \right)
& \text{if}\ 1/2\leq \gamma<2/\pi \,.
\end{cases}
  \end{align*}
\end{Thm}

\begin{proof}
The proof is the same as that of Theorem \ref{existencerhoh}, except that when computing $\tr C$ we now use Lemma \ref{auxbound} for $s>3/4$.
\end{proof}

%%%%%%%%%%%%%%%%%%%%%%

\appendix
\section{Auxiliary estimates}
\label{a:besselintegrals}

In this appendix we prove two bounds on integrals involving Bessel functions.

\begin{Lem}
  \label{auxboundnod}
  Let $M>0$ be some fixed constant and $s\in(1/2,1]$. Then for any
  $r\geq0$ and $\nu\geq1/2$,
  \begin{align*}
    & \int_0^\infty \dk\ \frac{kr J_{\nu}(kr)^2}{(\sqrt{k^2+1}-1+M)^{2s}}
    \leq A_{s,M}\left[\left(\frac{r}{\nu}\right)^{2s-1}\one_{\{r\leq\nu\}}+\one_{\{r> \nu\}}\right]\,.
  \end{align*}
\end{Lem}

\begin{proof}
  We start with the case $s<1$. We shall use the bounds
  \begin{align*}
    J_{\nu}(x)^2\leq \const \left( \frac{1}{\nu^2} \one_{\{ x\leq \nu\}} + \frac{1}{x} \one_{\{ x>\nu\}} \right).
  \end{align*}
  For $\nu\geq 1$ these bounds are stated in \cite[Lemma 3.2]{Cordoba2016}, even in a slightly stronger form where the different behavior happens at $x=3\nu/2$. This is stronger since $1/\nu^2 \leq 1/\nu\leq 3/(2x)$ for $\nu\leq x\leq 3\nu/2$ and $\nu\geq 1$. The above bounds continue to hold for $\nu\geq\nu_0$ for any $\nu_0>0$, in particular, for $\nu_0=1/2$. This follows from \cite[Equation 9.1.60]{Olver1968} and the fact that the asymptotics in \cite[Equation 9.2.1]{Olver1968} are uniform in $\nu\in[\nu_0,1]$.

Using these bounds, we can estimate the integral in the lemma by a constant times
  \begin{align}
    \label{eq:besselintegralsnod}
    \begin{split}
      & \nu^{-2}\int_0\limits^{\nu/r}\dk\ \frac{kr}{(\sqrt{k^2+1}-1+M)^{2s}} + \int\limits_{\nu/r}^\infty \dk\ \frac{1}{(\sqrt{k^2+1}-1+M)^{2s}}\,.
    \end{split}
  \end{align}
For the denominator we have $\sqrt{k^2+1}-1+M \gtrsim \one_{\{ k\leq 1 \}} + k \one_{\{k>1\}}$. Recalling $1/2<s<1$ we see that for $r\leq \nu$ the expression \eqref{eq:besselintegralsnod} is bounded by a constant times
  \begin{align*}
    & \left[\nu^{-2}\int_0^1 kr\,\dk + \nu^{-2}\int_{1}^{\nu/r}k^{1-2s}r\,\dk + \int_{\nu/r}^\infty k^{-2s}\,\dk\right]\one_{\{r\leq \nu\}}\\
    & \quad \leq A_s\left[\frac{r}{\nu^2}+\frac{r}{\nu^2} \left(\frac{r}{\nu}\right)^{2s-2}+\left(\frac{r}{\nu}\right)^{2s-1}\right]\one_{\{r\leq \nu\}}
    \leq A_s\left(\frac{r}{\nu}\right)^{2s-1}\one_{\{r\leq \nu\}}
  \end{align*}
  and for $r> \nu$ by a constant times
  \begin{align*}
    & \left[\nu^{-2}\int_0^{\nu/r}kr\,\dk + \int_{\nu/r}^1\dk + \int_1^\infty k^{-2s}\,\dk\right]\one_{\{r> \nu\}}\\
    & \quad \leq A_s\left[\frac{1}{r}+1\right]\one_{\{r> \nu\}}
      \leq A_s\one_{\{r> \nu\}}\,.
  \end{align*}
This proves the claimed bounds for $s<1$.

For $s=1$ it clearly suffices to bound the integral with the denominator replaced by $k^2+1$. The corresponding quantity is equal to
  \begin{align}
    \label{eq:besselnonrel}
    \int_0^\infty \dk\ \frac{kr J_{\nu}(kr)^2}{k^2+1}=rK_{\nu}(r)I_{\nu}(r)\,.
  \end{align}
This follows either from the formula for the Green's function in Sturm--Liouville theory and the equations satisfied by the Bessel functions, or by combining  \cite[Formula 10.22.69]{NIST:DLMF} and \cite[Formulas 9.6.3 and 9.6.4]{Olver1968}. By the pointwise bound (see, e.g., Iantchenko et al
  \cite[p. 185]{Iantchenkoetal1996})
  \begin{align}
  \label{eq:besselboundils}
    K_\nu(\nu x)I_\nu(\nu x) \leq \frac{9}{4\nu(1+x^2)^{1/2}} \quad \text{for all}\ \nu\geq\frac12\,,
  \end{align}
we obtain the claimed bound for $s=1$.
\end{proof}

\begin{Lem}
  \label{auxbound}
  Let $a>0$ and $s\in(1/2,3/4]$. Then, for any $r\geq0$ and $\nu\geq1/2$,
  \begin{align*}
    & \int_0^\infty \dk\ \frac{kr J_{\nu}(kr)^2}{(\sqrt{k^2+1}-1+a\nu^{-2})^{2s}}\\
    & \quad\leq A_{s,a}\left[\left(\frac{r}{\nu}\right)^{2s-1}\one_{\{r\leq\nu\}}+\left(\frac{r}{\nu}\right)^{4s-1}\one_{\{\nu \leq r\leq \nu^2\}} + \nu^{4s-1}\one_{\{r\geq \nu^2\}}\right]\,.
  \end{align*}
If $s\in(3/4,1]$, the bounds holds with the right side replaced by
$$
A_{s,a}\left[\left(\frac{r}{\nu}\right)^{\! 2s-1}\!\one_{\{r\leq\nu^\alpha\}}
+ \frac{r}{\nu^{4-4s}} \one_{\{\nu^\alpha< r\leq\nu^\beta\}}
+ \left(\frac{r}{\nu}\right)^{4s-1}\one_{\{\nu^\beta \leq r\leq \nu^2\}} + \nu^{4s-1}\one_{\{r\geq \nu^2\}}\right]
$$
and $\alpha=(5-6s)/(2-2s)$ and $\beta=(8s-5)/(4s-2)$ where, for $s=1$, we set $\nu^\alpha=0$.
\end{Lem}

\begin{proof}
We first assume $s<1$. We use the same bounds on $J_\nu^2$ as in the proof of Lemma \ref{auxboundnod} and find that the integral in the lemma is bounded by a constant times
  \begin{align}
    \label{eq:besselintegrals}
    \begin{split}
      & \nu^{-2}\int_0\limits^{\nu/r}\dk\ \frac{kr}{(\sqrt{k^2+1}-1+a \nu^{-2})^{2s}} + \int\limits_{\nu/r}^\infty \dk\ \frac{1}{(\sqrt{k^2+1}-1+a \nu^{-2})^{2s}}\,.
    \end{split}
  \end{align}
For the denominator we have  
$$
\sqrt{k^2+1}-1+a \nu^{-2} \gtrsim \nu^{-2} \one_{\{k\leq \nu^{-1}\}} + k^2 \one_{\{\nu^{-1} < k\leq 1\}} + k \one_{\{k> 1 \}}
$$

We bound \eqref{eq:besselintegrals} separately in the three regions appearing in the bound in the lemma. For $r>\nu^2$ it is bounded by a constant times
  \begin{align*}
    & \nu^{-2+4s}r\int_0^{\nu/r}\dk\, k + \nu^{4s}\int_{\nu/r}^{\nu^{-1}}\dk + \int_{\nu^{-1}}^1\dk\, k^{-4s} + \int_1^\infty \dk\, k^{-2s} \leq A_s \nu^{4s-1} \,,
  \end{align*}
for $\nu<r\leq\nu^2$ by a constant times
  \begin{align*}
    & \nu^{-2+4s}r\int_0^{\nu^{-1}}\!\!\dk\,k + \nu^{-2}r\int_{\nu^{-1}}^{\nu/r}\! \dk\,k^{1-4s} + \int\limits_{\nu/r}^{1}\dk\,k^{-4s} + \int_1^\infty\! \dk\,k^{-2s} \\
   & \qquad \leq A_s \left( \left(\frac{r}{\nu}\right)^{4s-1} + \frac{r}{\nu^{4-4s}} \right)
  \end{align*}
and for $r\leq\nu$ by a constant times
 \begin{align*}
  & \nu^{-2+4s}r\int_0^{\nu^{-1}}\dk\,k + \nu^{-2}r\int_{\nu^{-1}}^1 \dk\,k^{1-4s} + \nu^{-2}r\int_1^{\nu/r}\dk\,k^{1-2s} + \int\limits_{\nu/r}^\infty \dk\ k^{-2s} \\
  & \qquad \leq A_s \left( \left(\frac{r}{\nu}\right)^{2s-1} + \frac{r}{\nu^{4-4s}} \right).
 \end{align*}
We have $(r/\nu)^{4s-1}\leq r/\nu^{4-4s}$ if and only if $r\leq \nu^\beta$, and $(r/\nu)^{2s-1}\leq r/\nu^{4-4s}$ if and only if $r\geq \nu^\alpha$. Moreover, we have $\beta\leq 1\leq\alpha$ for $s\leq 3/4$, so in this case the term $r/\nu^{4-4s}$ can be dropped. This proves the claimed bound for $s<1$.

For $s=1$ the same argument works, except that in case $r\leq\nu$, the integral over the region $1\leq k\leq\nu/r$ gives an additional logarithm. Therefore we bound the corresponding integral by
$$
\int_1^\infty \dk\ \frac{kr J_{\nu}(kr)^2}{(\sqrt{k^2+1}-1+a\nu^{-2})^{2s}}
\lesssim \int_0^\infty \dk\ \frac{kr J_{\nu}(kr)^2}{k^2+1} \lesssim \frac{r}{\nu} \,.
$$
The last bound follows from \eqref{eq:besselnonrel} and \eqref{eq:besselboundils}, recalling that $r\leq\nu$. This gives the desired bound for $s=1$.
\end{proof}

%%%%%%%%%%%%%%

\section{A bound on the non-relativistic hydrogenic density}
\label{a:nonrelrhoh}

As remarked in the introduction, Heilmann and Lieb \cite{HeilmannLieb1995} showed
that the non-relativistic hydrogenic density behaves like
$(\sqrt{2}/(3\pi^2)) \gamma^{3/2} r^{-3/2}+o(r^{-3/2})$ as $r\to\infty$. 
Their proof relied on
the precise asymptotics of the eigenfunctions of the hydrogen operator.
Here we show that it is possible to adapt the arguments of the proof of Theorem
\ref{existencerhoh} to show that the non-relativistic hydrogenic density is bounded from above by a constant times $r^{-3/2}$.

Since no relativistic operator is going to appear in this appendix, we abuse notation and write $d_\ell$, $\rho_\ell^H$ and $\rho^H$ for the non-relativistic densities corresponding to the operator $-\Delta/2 -\gamma |x|^{-1}$. By scaling we may assume $\gamma=1$. As before, we write
$$
\rho_\ell^H(r) = \tr d_\ell \delta_r = \tr ABCB^*A^*
$$
where now
\begin{align*}
  A&:=d_{\ell}(p_\ell^2/2-1/r+b_\ell)^{1/2}\\
  B&:=(p_\ell^2/2-1/r+b_\ell)^{-1/2}(p_\ell^2/2+b_\ell)^{1/2}\\
  C&:=(p_\ell^2/2+b_\ell)^{-1/2}\delta_r(p_\ell^2/2+b_\ell)^{-1/2}
\end{align*}
with $b_\ell>-\inf\spec(p_\ell^2/2-1/r)=1/(2(\ell+1)^2)$. We take $b_\ell=a(\ell+1/2)^{-2}$ with a constant $a$ to be chosen later.

First, we have $\|A\|^2 \leq a(\ell+1/2)^{-2}$.

The norm of $B$ is bounded uniformly in $\ell$, since
\begin{align*}
  \frac{p_\ell^2}{2}-\frac{1}{r}+b_\ell = \frac12\left(\frac{p_\ell^2}{2}+b_\ell\right) + \frac12\left(\frac{p_\ell^2}{2}-\frac{2}{r}+b_\ell\right)
\end{align*}
and the right side is bounded from below by $(p_\ell^2/2+b_\ell)/2$ provided $a>0$ is chosen large enough such that also $p_\ell^2/2-2/r+b_\ell\geq 0$.

Finally, as in \eqref{eq:besselnonrel} and \eqref{eq:besselboundils},
\begin{align*}
\tr C & =  (p_\ell^2/2+a_\ell)^{-1}(r,r) \\
& = \int_0^\infty \dk\ \frac{kr J_{\ell+1/2}(kr)^2}{k^2/2+a(\ell+1/2)^{-2}}\\
                       & = 2r K_{\ell+1/2}\left(\frac{\sqrt{2a}\, r}{\ell+1/2}\right) I_{\ell+1/2}\left(\frac{\sqrt{2a}\, r}{\ell+1/2}\right) \\
    & \leq A_a \left[\frac{r}{\ell+1/2}\, \one_{\{r\leq(\ell+1/2)^2\}} + (\ell+1/2) \,\one_{\{r>(\ell+1/2)^2\}}\right].
\end{align*}  

To summarize, we have shown that for all $\ell\in\N_0$,
\begin{align*}
\tr ABCB^*A^* & \leq \|A\|^2 \|B\|^2 \tr C \\
& \leq A_a(\ell+1/2)^{-2}\left[\frac{r}{\ell+1/2}\one_{\{r\leq(\ell+1/2)^2\}} + (\ell+1/2)\one_{\{r>(\ell+1/2)^2\}}\right].
\end{align*}
This implies
\begin{align*}
\rho^H(r) = r^{-2} \sum_{\ell=0}^\infty (2\ell+1)\rho_\ell^H(r)
  \leq A_{a}\left(r^{-1} \one_{\{r\leq 1\}} + r^{-3/2} \one_{\{r>1\}} \right) ,
\end{align*}
which is the claimed bound.
\qed

%%%%%%%%%%%%%

\section*{Acknowledgments}

The authors acknowledge partial support by the U.S. National
Science Foundation through grants DMS-1363432 (R.L.F.)
and DMS-1665526 (B.S.), by the Deutsche Forschungsgemeinschaft (DFG, German Research Foundation) through
grant SI 348/15-1 (H.S.) and through Germany's Excellence Strategy – EXC-2111 – 390814868 (R.L.F., K.M., H.S.) and by the Israeli BSF through grant No. 2014337 (B.S.).

\bibliographystyle{plain}

%\bibliography{coulomb}

\end{document}